\pdfoutput=1

\documentclass[sigplan,10pt]{acmart}
\renewcommand\footnotetextcopyrightpermission[1]{}

  \usepackage{amsmath,amssymb,amsfonts}
  \usepackage{algorithmic}
  \usepackage{graphicx}
  \usepackage{textcomp}
  \usepackage{xcolor}
  \usepackage{xspace}
  \usepackage{enumerate}
  \usepackage{lscape}
  \usepackage{hhline}
  \usepackage{algorithm}
  \usepackage{amsthm}
  \usepackage{setspace}  
  \usepackage{caption}

\newtheorem{lemma}{\textbf{Lemma}}
\newtheorem{theorem}{\textbf{Theorem}} 
 \newtheorem{example}{\textbf{Example}}
\newtheorem{definition}{\textbf{Definition}}

\newcommand{\TaskShuffler}{{\textsc{TaskShuffler}}}
\newcommand{\TS}{\TaskShuffler}
\newcommand{\NewTS}{{\textsc{TaskShuffler++}}} 
\newcommand\revised[1]{{\color{black} {#1}}}

\begin{document}
\settopmatter{printacmref=false} 
\renewcommand\footnotetextcopyrightpermission[1]{} 
\pagestyle{plain}

\title{\huge TaskShuffler++: Real-Time Schedule Randomization for Reducing Worst-Case Vulnerability to Timing Inference Attacks}

\author{Man-Ki Yoon}
\affiliation{%
  \institution{Yale University}  
}

\author{Jung-Eun Kim}
\affiliation{%
  \institution{Yale University}  
}

\author{Richard Bradford}
\affiliation{%
  \institution{Collins Aerospace}  
}

\author{Zhong Shao}
\affiliation{%
  \institution{Yale University}  
}

\begin{abstract}
  This paper presents a schedule randomization algorithm that reduces 
  the vulnerability of real-time systems to timing inference attacks which attempt to learn the timing of task execution. It utilizes run-time information readily available at each scheduling decision point to increase the level of uncertainty in task schedules, while preserving the original schedulability. The randomization algorithm significantly reduces an adversary's best chance to correctly predict what tasks would run at arbitrary times. This paper also proposes an information-theoretic measure that can quantify the worst-case vulnerability, from the defender's perspective, of an arbitrary real-time schedule.
\end{abstract}

\maketitle
\section{Introduction} 
\label{sec:intro}

Time is a lethal source of information leakage. Adversaries can exploit inherent temporal characteristics of certain algorithmic operations to extract sensitive information such as cryptographic keys~\cite{Ge2018,PereidaGarcia:2016,Liu:2015:LLC} or to inject false data for system destabilization \cite{SchedLeakRTAS2019,Nasri2019}. A critical requirement for successful launch of such attacks is to pinpoint the \emph{timing} of the target's execution. By narrowing the time range when the target would execute, an adversary can initiate and operate an attack at the right time. Such attempts become more effective if the adversary can predict the timing as precisely as possible.

Real-time systems 
are particularly vulnerable to such \emph{timing attacks} because of their \emph{deterministic} nature in operation. In particular, the limited uncertainty in the repeating schedules of real-time tasks, attained by real-time scheduling, enables inference of task execution timings~\cite{SchedLeakRTAS2019}, and thus provides a ground for potential security risks (e.g., false-data injection, side-channel attack) \cite{Nasri2019}. This predictability also helps adversaries establish a covert timing-channel~\cite{embeddedsecurity:volp2008, embeddedsecurity:son2006} through which they can communicate indirectly by altering execution behavior.  

Hence, scheduling obfuscation methods have been proposed~\cite{YoonTaskShuffler,Kruger2018ECRTS} as solutions to timing inference attacks to real-time systems. Scheduling obfuscation can be classified as a moving target defense (MTD) technique~\cite{MTDDHS} which aims to make it more difficult for adversaries to launch a successful security attack by dynamically changing the attack surface~\cite{davi2013gadge,crane2015readactor}. In particular, \TS{}~\cite{YoonTaskShuffler} reduces the determinism perceived by adversaries in fixed-priority real-time schedules by randomly allowing \emph{priority inversions}, subject to schedulability constraints. 

\TS{}, however, has a fundamental limitation in that task executions can have \emph{temporal locality} due to (i) the pessimism in the off-line analysis of tasks' budgets for priority inversions and (ii) the 
job selection process that does not consider any timing properties of tasks. This makes it easier for an adversary to predict the timing of certain tasks \cite{Nasri2019}. Hence, in this paper we propose a new schedule randomization algorithm, which we call \NewTS{}, that overcomes this limitation. It is based upon the same principles of the randomization mechanism -- namely, priority inversions. However, it does not rely on the static bounds ocn the priority inversion budgets. Instead, it makes use of run-time information that is readily available to the scheduler to determine which tasks are allowed to execute (thus the level of priority inversion) at a given time, while preserving the system's schedulability. This enables an added level of randomness by accurately reflecting the possibilities in the present situation. The uncertainty is further increased in \NewTS{} by treating each of the candidate jobs differently in the final selection process by giving different \emph{weights} to different candidates. As will be seen later, this weighted job selection in fact decreases the chance that tasks have temporal locality by spreading executions across a wider range.

More importantly, \NewTS{} is designed to reduce \emph{schedule min-entropy}, which we propose as a new measure of uncertainty introduced by \emph{any} randomization technique to schedules. It quantifies how \emph{vulnerable} a schedule of a task set is to an adversary's correct prediction in the \emph{worst-case} from the defender's perspective. 
We argue that the schedule min-entropy is the proper indicator for the level of security against timing inference attacks and also that \NewTS{} achieves significantly increased uncertainty in execution timing.  

In summary, this paper makes the following contributions:
\begin{enumerate}
\item We introduce \NewTS{}, a new schedule randomization algorithm that achieves increased uncertainty by taking into account run-time information in the candidate selection as well as the final job selection process;
\item An approximate algorithm that implements \NewTS{} with reduced exactness but with lower complexity; and 
\item A new metric to measure the worst-case vulnerability of a randomized schedule to a timing inference attack.
\end{enumerate}
\section{Problem Description}

\subsection{System Model}\label{subsec:system_model}

We consider a set of $N$ sporadic tasks $\Gamma = \{\tau_1, \tau_2, \ldots, \tau_N\}$ on a uniprocessor system. Each task $\tau_i$ is
characterized by $(e_i, p_i, d_i)$, where $e_i$ is the worst-case execution
time, $p_i$ is the minimum inter-arrival time between successive
releases and $d_i$ is the relative deadline. We assume implicit deadline, i.e., $d_i=p_i$. \revised{Each invocation of task, i.e., job, may execute for an arbitrary amount of time, upper-bouned by the WCET, $e_i$. But, the actual execution time is even unknown to the scheduler. }
Task priorities are fixed (e.g., assigned according to Rate Monotonic (RM)~\cite{LiuLayland1973}) and distinct. 
Let $hp(\tau_i)$ denote the set of tasks that have higher priorities than $\tau_i$ and $lp(\tau_i)$ denote the set of tasks with lower priorities. Due to the assumption of implicit deadlines, there can be at most one job per task at any time instant if the task is schedulable. Hence we use the terms, task and job, interchangeably and denote each by $\tau_i$.  When no ambiguity arises, the task set $\Gamma$ includes the idle task $\tau_{I}$ that has the lowest priority and $e_I = L - \sum_{\tau_i\in \Gamma} e_i$ and $p_I = L$ where $L$ is length of the hyper-period 
 of $\Gamma$ (i.e., the least common multiple of the task periods).  

We consider a task set $\Gamma$ that is \emph{schedulable} by a
fixed-priority preemptive scheduling. That is, the worst-case
response time of task $i$ 
is less than or equal to the deadline, $d_i$, 
and it is calculated by the iterative response time analysis~\cite{Audsley93}: 
\vspace{-0.1cm}
\begin{align}
\label{eq:wcrt_i}
w_i^{k+1} = e_i + \sum_{\tau_j \in hp(\tau_i)} \bigg \lceil \frac{w_i^{k}}{p_j} \bigg \rceil e_j \le d_i,  	
\end{align} 
where $w_i^0 = e_i$ and the worst-case response time is $w_i^{k+1} = w_i^{k}$ for some $k$. Finally, we assume that there is no synchronization or precedence constraints among tasks and that
$e_i, p_i, d_i \in \mathbb{N}^+$.

The presence of release jitters (in moderation) can increase the randomness of schedules and thus reduce the vulnerability to timing attacks \cite{YoonTaskShuffler}. It is straightforward to extend our analysis to incorporate jitters, 
but to simplify the discussion, we assume no release jitters.  

\subsection{Problem Description}

\noindent\textbf{Adversary model: }
We consider an adversary that tries to infer task execution patterns from observations (e.g., via a form of side-channel) collected over a certain period of time. We do not assume a particular type of side-channel; instead, we consider any form of observation that can help infer the timing of task executions. Such an attack is possible especially in real-time systems because task executions repeat with slight variations. SchedLeak~\cite{SchedLeakRTAS2019} demonstrated such an attack - the attacker task, which knows a victim task's period, observes its own executions, i.e., when it is scheduled and not, over many hyper-periods to infer the victim task's timing. 
Nasri et al. in \cite{Nasri2019} discuss various schedule-based attacks that can be facilitated by successful timing inference.

A schedule randomization raises the time complexity of such attempts because the attacker will observe non-deterministic schedules over time. However, this still requires an assumption that the information on the scheduler's state such as the ready queue is not available to the attacker and thus that the attacker cannot perform a timing inference \emph{instantaneously}. 
\revised{
Nasri et al. in \cite{Nasri2019} analyzes in detail the vulnerability of schedule randomization techniques when the attacker is able to directly observe the scheduler's state. In particular, by observing what tasks have been scheduled and when up to the present and knowing their remaining execution budgets precisely, the attacker can pinpoint the victim task's timing. The predictions on the future schedule of certain tasks become especially easier towards their deadline due to the schedulability constraints. Therefore, we consider the \emph{Weak} attacker model defined in \cite{Nasri2019} which is considered also in \cite{Kruger2018ECRTS, SchedLeakRTAS2019}; that is, attackers cannot observe scheduler state -- what tasks have been scheduled and what tasks are in the ready queue. An attacker can only observe its own execution intervals. SchedLeak~\cite{SchedLeakRTAS2019} demonstrated that, although the attacker cannot help but spend a non-trivial amount of time (i.e., multiple hyper-periods), he/she can infer the timing of a victim task by only observing his/her execution.  

Lastly, we assume that the scheduler is trustworthy and protected from the adversary. Otherwise, the attacker can observe task schedules directly no matter how they are scheduled or even perform more active attacks than timing inference. }

\vspace{0.5\baselineskip}
\noindent\textbf{Goal: } For a task set $\Gamma$ that is schedulable by a fixed-priority scheduling without any schedule randomization, our goal is to randomize its schedule in the run-time for an increased difficulty for an adversary to predict which of $\Gamma$ would run at an arbitrary time slot $t$, while guaranteeing the schedulability of the tasks.

\NewTS{} algorithm consists of two phases: (i) \emph{candidate selection} and (ii) \emph{job selection}. It first forms a list of candidate jobs that are allowed to execute at a particular time slot and then selects one from the list in a non-deterministic manner. Most importantly, these processes must be designed in such a way that the task set is \emph{still} schedulable even when randomized.

\subsection{Background: \TaskShuffler{}}
\label{sec:TaskShuffler}

In this section, we briefly review the \TaskShuffler{} algorithm~\cite{YoonTaskShuffler}. It randomizes the schedule of a fixed-priority task set by allowing a random \emph{priority inversion} at each point where a scheduling decision is to be made. 
Since arbitrary priority inversions may cause deadline misses of high priority tasks, 
\TaskShuffler{} limits the amount of priority inversion that each task can endure. For the enforcement, the \emph{worst-case maximum inversion budget} $V_i$ for each task $\tau_i$ is calculated off-line. The budget decreases as $\tau_i$ is delayed by priority inversions by any of lower priority tasks $lp(\tau_i)$. When the budget becomes zero, none of $lp(\tau_i)$ is allowed to execute until $\tau_i$ finishes its current job. The budget is replenished to $V_i$ when a new job of $\tau_i$ is released. 
\begin{figure}[t]
\centering
\includegraphics[width=1\columnwidth]{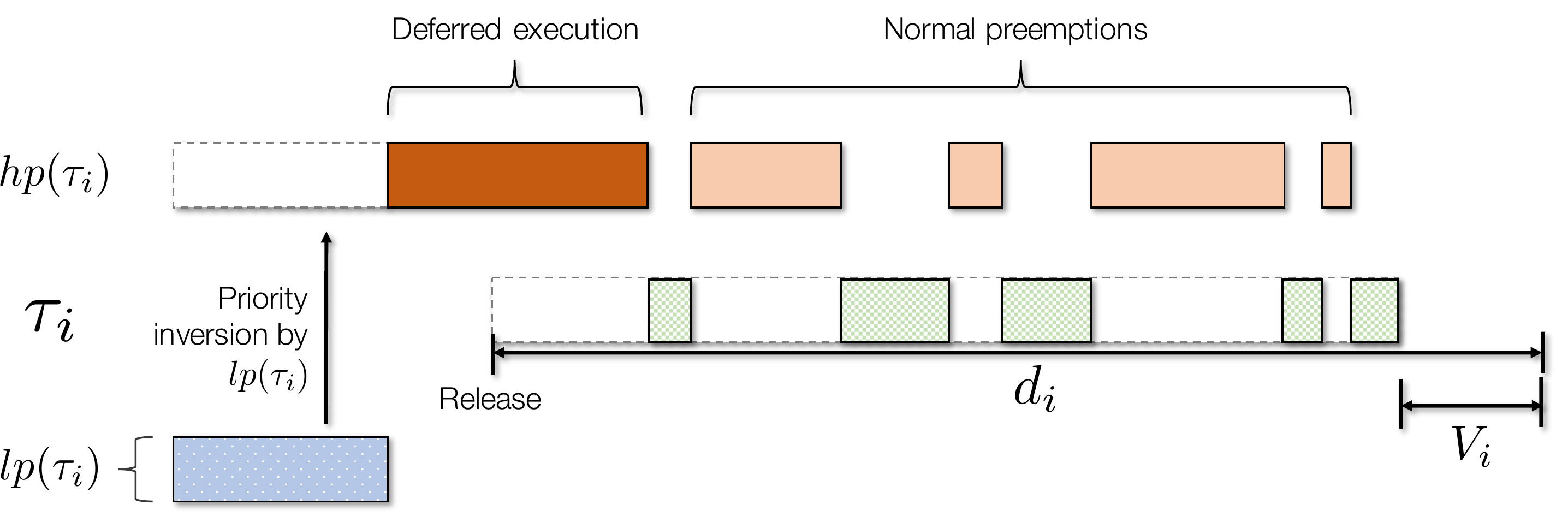}
\caption{Worst-case maximum inversion budget in {\TaskShuffler{}}~\cite{YoonTaskShuffler}.}\label{fig:TaskShufflerMaximumBudget}
\end{figure}
It
is calculated using the worst-case interference from the higher priority tasks $hp(\tau_i)$. In the presence of arbitrary priority inversions, $\tau_i$ can experience more interference from $hp(\tau_i)$ than when no priority inversion is allowed, due to the additional interference by the \emph{deferred executions}~\cite{Rajkumar:1988,Audsley93} caused by some of $lp(\tau_i)$ when $\tau_i$ is not running, as illustrated in Figure~\ref{fig:TaskShufflerMaximumBudget}. 
In addition, to prevent a task $\tau_x$ with $V_x<0$ from missing its deadline, a run-time policy called \emph{Level-$\tau_x$ exclusion policy} in \TaskShuffler{} enforces that none of $lp(\tau_x)$ is allowed to run while any of $hp(\tau_x)$ has an unfinished job, thus removing the deferred executions. 
For the details of the \TS{} algorithm, the readers are referred to \cite{YoonTaskShuffler}.
\section{\NewTS{}}\label{sec:NewTS}

We first explain the limitations of \TS{} and then present a new schedule randomization technique that we call \NewTS{} and discuss how it solves the problems of \TS{}.

\subsection{Problems of \TS{}}\label{subsec:prob_TS}

\TS{} has two major limitations that lead to high \emph{temporal locality}; some tasks appear in a particular time period with higher probability, which makes the schedule more vulnerable to timing inference attack. First of all, the maximum inversion budget $V_i$ is not tight since it is calculated \emph{statically} assuming that deferred executions can happen every time. This lowers the maximum budget $V_i$, which reduces the possibilities that $\tau_i$ could be shuffled with $lp(\tau_i)$. Also, such a lower budget activates the level-$\tau_x$ exclusion policy more frequently than necessary, further reducing priority inversions by certain low priority tasks. Let us consider Example~\ref{example:prob_TS}.
\begin{example}\label{example:prob_TS}
Consider the following task set. Each ${V}_i$ is calculated as shown in the last column of the table:
\begin{center}\small
\begin{tabular}{|c||c|c|c||c|}
\hline 
 & $p_i$ & $e_i$ & $d_i$ & $V_i$ \\ 
\hline \hline 
$\tau_1$ & 5 & 2 & 5 & 3 \\ 
\hline 
$\tau_2$ & 7 & 2 & 7 & -1 \\ 
\hline 
$\tau_3$ & 20 & 3 & 20 & -1 \\ 
\hline 
\end{tabular} 
\end{center}

\begin{center}
\includegraphics[width=0.97\columnwidth]{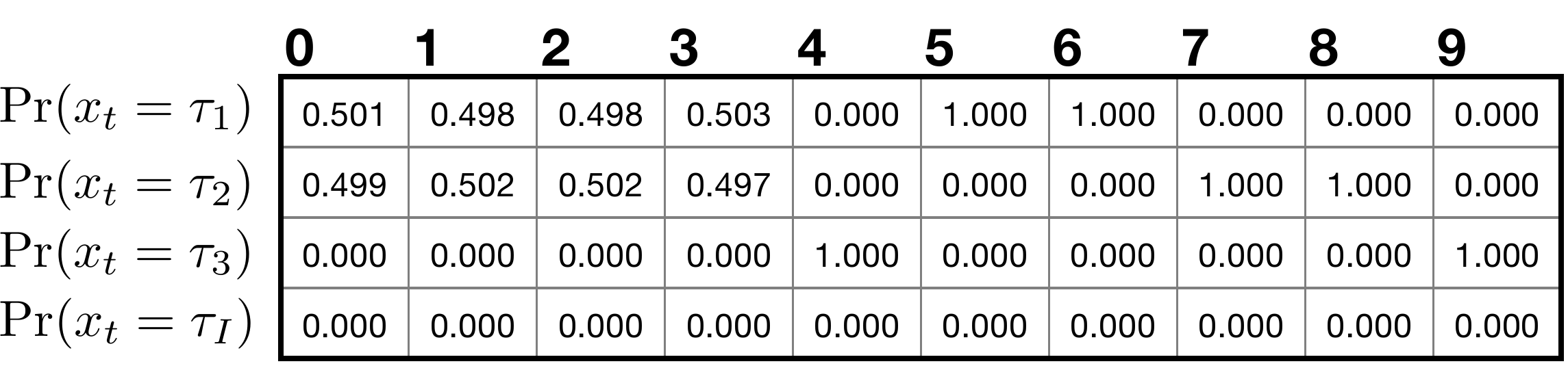}
\end{center}

The table above shows the probability of seeing each task in the first 10 time slots in each hyper-period, measured by running the task set for 100 thousands hyper-periods under \TS{}. 
\end{example}
In the example above, $\tau_3$ cannot execute during $[0,3]$ because $\tau_2$ cannot have a positive inversion budget. Because of this, both $\tau_1$ and $\tau_2$ released at $t=0$ always finish by $t=4$. This leads $\tau_1$ to appear at time $5$ and $6$ in a \emph{deterministic} manner. For a similar reason, $\tau_2$'s second release always appear at time $7$ and $8$. If the attacker was able to observe the scheduler's state directly, he/she can precisely target these timings \cite{Nasri2019}. As we will see later, these tasks can be shuffled in a more non-deterministic way by not using the static (thus pessimistic) bounds on the 
 inversion budget. 

Another limitation of \TS{} is in the final job selection from the candidate list. When there are $n$ candidate jobs, each of them can be selected with an equal probability $1/n$. This uniform-weighting further increases the temporal locality of certain jobs by \emph{not} considering the load characteristics. Section~\ref{subsec:weighted_selection} discusses the problem in detail.

\subsection{Design of \NewTS{}}

\noindent\textbf{Dynamic Budgeting for Candidate Selection: }
The pessimism in computing the bounds on the maximum inversion budget in \TS{} is due to the nature of the static analysis that needs to apply the worst-case scenario to any situation. Hence, \NewTS{} solves the problem by taking a \emph{dynamic} approach. That is, instead of enforcing a fixed budget for priority inversion, a decision is made at run-time about whether to allow a priority inversion by each job in the ready queue. If it is certain that the priority inversion at present would not cause any deadline miss of the higher priority tasks in the future with the worst-case scenario, the job is added to the candidate list. As will be discussed in detail in Section~\ref{subsec:exact_newTS}, this requires an \emph{iterative} computation, which is not suitable for use during run-time. Hence, we also present an \emph{approximation} algorithm for \NewTS{} (in Section~\ref{subsec:approx_newTS}) that sacrifices the tightness of the inversion budget bound for lower complexity. 

\vspace{0.5\baselineskip}
\noindent\textbf{Weighted Random Job Selection: }
Given a set of candidate jobs at each time, \NewTS{} selects one non-deterministically while taking into account each job's \emph{remaining workload} until the deadline, which reflects the job's urgency. By assigning a higher weight to a more urgent job now, we reduce the chance that it becomes more urgent later at which point the job would appear with even higher probability. As will be seen in Section~\ref{subsec:weighted_selection}, this weighted random job selection in fact reduces the chance of biased appearance.

\newcommand{\ReadySet}{\mathcal{R}_t}    

\subsection{\NewTS{} Algorithm}\label{subsec:exact_newTS}

Let $L_{\ReadySet} = (\tau_{(1)}, \tau_{(2)}, \ldots, \tau_{|\ReadySet|})$ be the ready queue of the jobs sorted in decreasing order of priority at time $t$. A schedule randomization algorithm picks a job from $L_{\ReadySet}$ in a non-deterministic way instead of selecting the highest priority job $\tau_{(1)}$. As mentioned in Section~\ref{sec:TaskShuffler}, \TS{} makes a random selection based on the worst-case maximum inversion budgets calculated offline. \NewTS{}'s main difference is that it does not rely on such a fixed budget. It instead analyzes each ready job based on the current state and the future executions of the higher priority tasks. 

Suppose we are to make a scheduling decision at time $t$. Then, for each job $\tau_{(i)} \in L_{\ReadySet}$, starting from the highest priority, the scheduler tests if $\tau_{(i)}$'s execution at time $t$ would not cause any deadline miss for all of the higher priority tasks $hp(\tau_{(i)})$ whether or not they are \emph{active} (i.e., has an outstanding job at present). If a priority inversion by $\tau_{(i)}$ may lead any task in $hp(\tau_{(i)})$ to miss deadline, $\tau_{(i)}$ cannot be added to the candidate list for time $t$.

\begin{figure}[t]
\centering
\includegraphics[width=1\columnwidth]{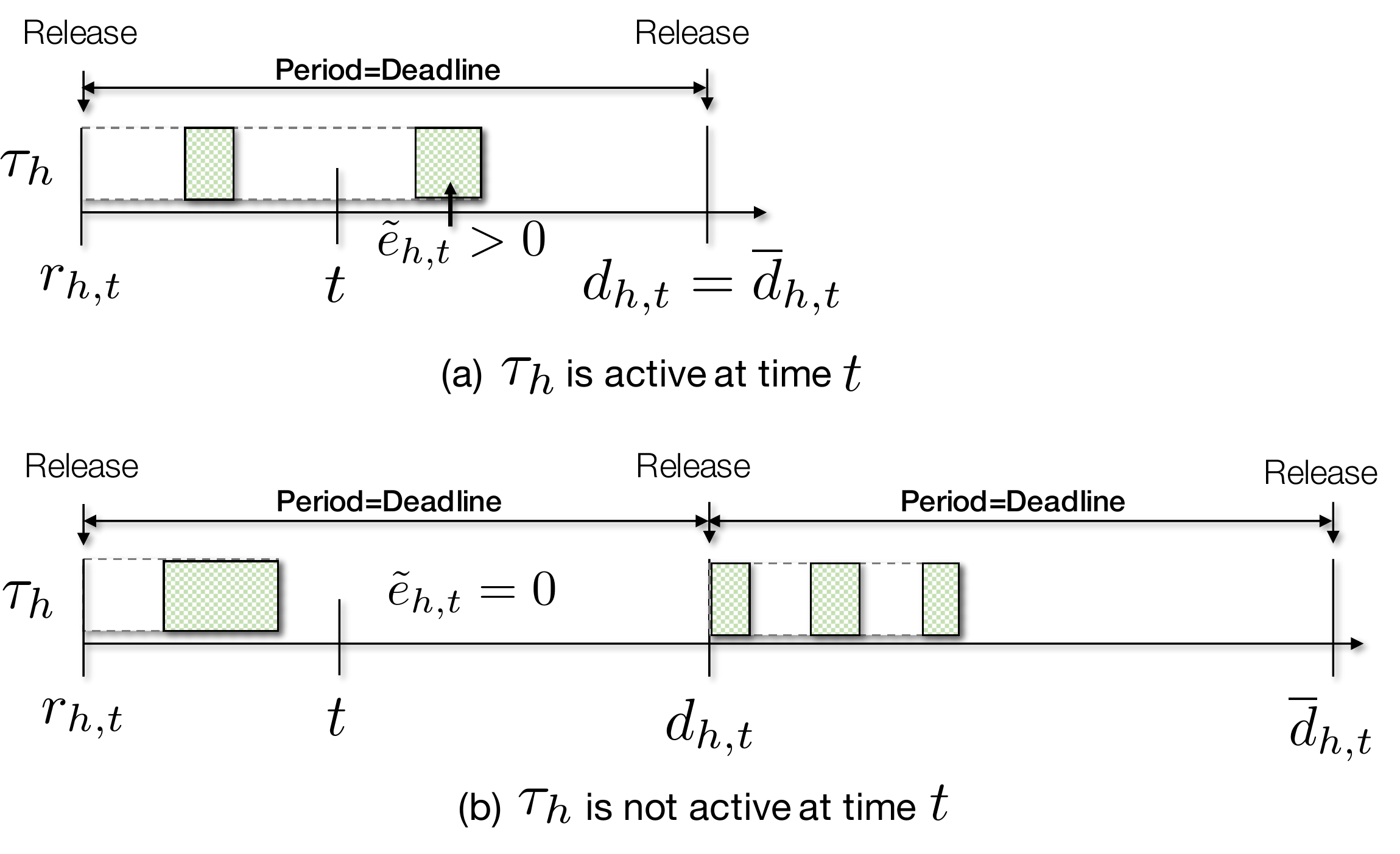}
\caption{The last release time $r_{h,t}$, the residue execution $\tilde{e}_{h,t}$, the current deadline $d_{h,t}$, and the effective deadline $\overline{d}_{h,t}$ of $\tau_h$ as of $t$.}
\label{fig:SomeNotions}
\end{figure}

In order to formally describe the candidate selection process in \NewTS{}, let us first introduce a set of definitions and notations,  which are also depicted in Figure~\ref{fig:SomeNotions}, that characterizes the state of a task at a particular time instant. 
\begin{definition}[$r_{h,t}$: Last release time of task $h$ as of time $t$] 
\label{def:r_h_t}
$r_{h,t}$ is the most recent release time of $\tau_h$ before time $t$. 
\end{definition}

\begin{definition}[$d_{h,t}$: Current deadline of task $h$ as of time $t$] 
\label{def:d_h_t}
$d_{h,t}$ is the deadline of the current job (i.e., most recently released) of $\tau_h$ released at $r_{h,t}$, which is computed by $d_{h,t} = r_{h,t} + d_h$, where $d_h$ is the deadline of $\tau_h$. 
If $\tau_h$ is active at time $t$, $D_{h,t}$ is the absolute deadline of the current release. If $\tau_h$ is not active, $D_{h,t}$ is the absolute deadline of the upcoming release, assuming that it arrives with the minimum inter-arrival time of $\tau_h$. 
\end{definition}

\begin{definition}[$\tilde{e}_{h,t}$: Residue execution of task $h$ at time $t$] 
\label{def:residue_e}
$\tilde{e}_{h,t}$ is the amount of outstanding execution of $\tau_h$ as of time $t$. $\tilde{e}_{h,t}> 0$ if $\tau_h$ is active and thus in the ready queue at time $t$. Otherwise (i.e. $\tau_h$ was completed before $t$) $\tilde{e}_{h,t}=0$. \revised{The scheduler computes $\tilde{e}_{h,t}$ by subtracting the amount of time it has executed since its arrival from the worst-case execution time, $e_h$.}
\end{definition}

\begin{definition}[$\overline{d}_{h,t}$: Effective deadline of task $h$ as of time $t$] 
\label{def:effective_d_h_t}
$\overline{d}_{h,t}$ is the absolute deadline of its current release if it is active at $t$: 
\[\overline{d}_{h,t} = d_{h,t} = r_{h,t} + d_h, \]
or that of the upcoming release if it is inactive at time $t$, assuming that it would arrive with its minimum inter-arrival time:
\[\overline{d}_{h,t} = r_{h,t} + p_h + d_h = r_{h,t} + 2d_h. \]

\end{definition}

\vspace{0.3\baselineskip}
\noindent\textbf{Schedulability test of $\tau_h$:} In what follows, we will test whether $\tau_h \in hp(\tau_{(i)})$ would miss its (effective) deadline if a candidate $\tau_{(i)}$ executes at time $t$ and thus a priority inversion occurs. As will be explained, it is important to make a distinction between the following two cases: (i) when $\tau_h$ is active and (ii) when $\tau_h$ is not active at time $t$ (see Figure~\ref{fig:SomeNotions}). In the former case, the active job of $\tau_h$ should complete by $\overline{d}_{h,t} = d_{h,t}$. In the latter case, we test if the \emph{upcoming} job of $\tau_h$ that will arrive at $r_{h,t} + p_h$ would finish by $\overline{d}_{h,t} =  r_{h,t} + 2 d_h$. The need for considering the schedulability of a future job of $\tau_h$ is because a priority inversion at time $t$ can indirectly interfere (i.e., deferred executions of $hp(\tau_h)$ as explained in Section~\ref{sec:TaskShuffler}) with the future job of $\tau_h$ releasing at $r_{h,t} + p_{h}$. The priority inversion is not allowed if it can lead to a deadline miss of the future $\tau_h$. In order to check for this, we need to compute the \emph{worst-case busy interval} that begins with a priority inversion by a low-priority job in $lp(\tau_h)$ at time $t$. 

\begin{figure}[t]
\centering
\includegraphics[width=1\columnwidth]{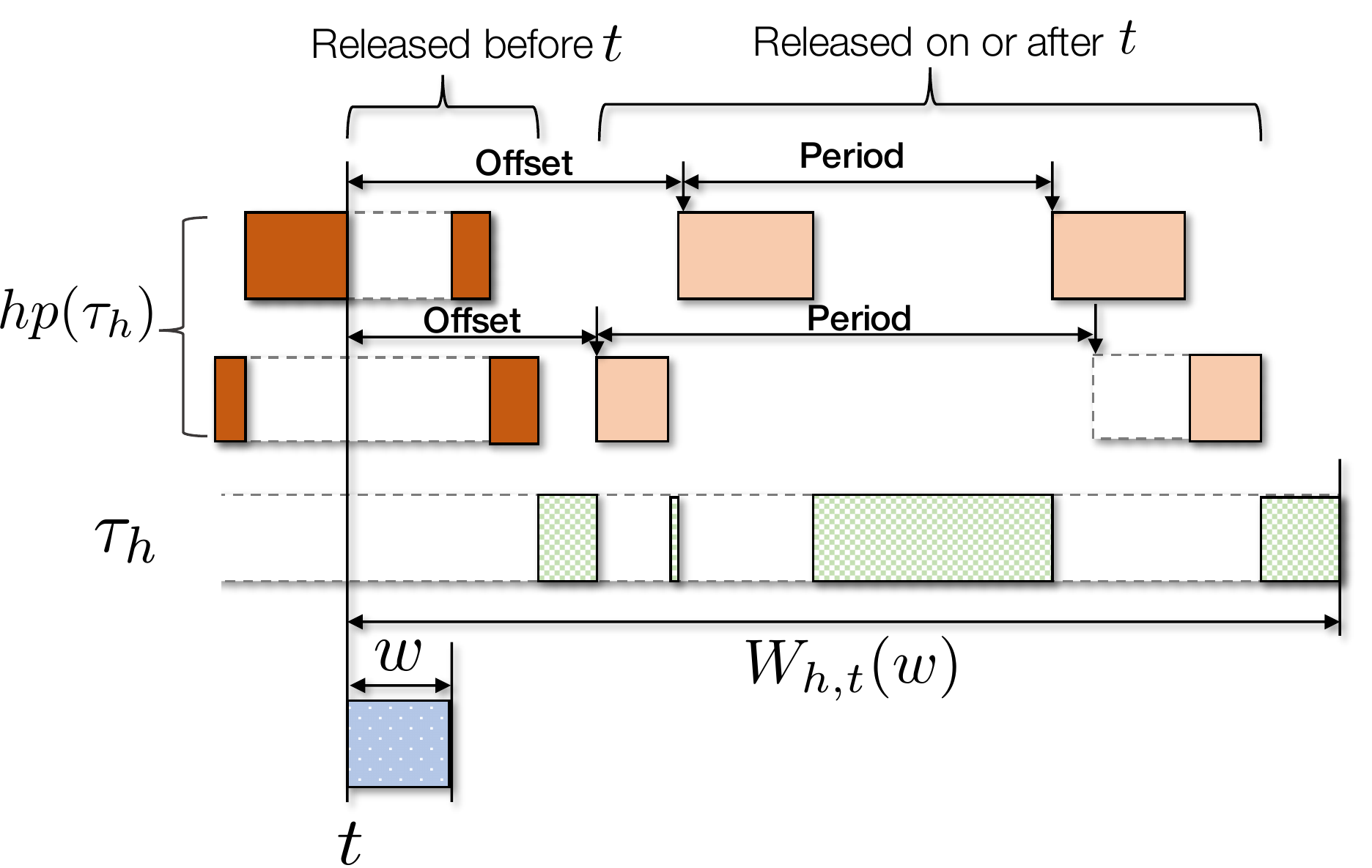}
\caption{Level-$\tau_h$ busy interval with base time $t$ with initial window of size $w$, when $\tau_h$ is active at time $t$.}\label{fig:busy_period}
\end{figure}

\begin{definition}
\label{def:busy_interval}
The level-$\tau_h$ busy interval with base time $t$ and initial window of size $w$, denoted by $W_{h, t}(w) $ and shown in Figure~\ref{fig:busy_period}, is a time window $[t, t+q]$ that is comprised of: 
\begin{enumerate}
\item a priority inversion of size $w$ by a job in $lp(\tau_h)$, 
\item all residue executions of $hp(\tau_{h})$ as of time $t$,
\item the residue execution of $\tau_{h}$ if active at time $t$, and 
\item all the future jobs of $hp(\tau_{h})$ that will arrive on or after $t$ and complete by the end of the busy interval.
\end{enumerate}
$W_{h, t}(w) = q$ is the length of the busy interval and $t+W_{h, t}(w)$ is the first time instant when all jobs of $hp(\tau_{h})$ released during $[t, W_{h, t}(w)]$ and $\tau_{h}$ itself, if it has an active job during the window, are complete. 
\end{definition}
$W_{h, t}(w)$ for given $t$ and $w$ is computed iteratively. Note that at time $t$, the amount of outstanding jobs (i.e., residue executions) of $hp(\tau_h)$ is known. The worst-case busy interval is when all the \emph{future} jobs of $hp(\tau_{h})$ released on or after $t$ arrive with their minimum inter-arrival times. Since the last release times of $hp(\tau_{h})$ before $t$ are known at time $t$, the relative \emph{offsets} of their next releases from time $t$ are known 
and calculated by $o_{j,t} = r_{j,t} + p_j - t$ for each $\tau_j\in hp(\tau_h)$. 
Then, $W_{h, t}(w)$ is computed by the following iterative equation:
\begin{align}\label{eq:busy_period}
W_{h,t}^{k+1}(w) = W_{h,t}^{0}(w) + \hspace{-0.3cm}\sum_{\tau_j\in hpe(\tau_h) } \Bigg\lceil \frac{W_{h,t}^{k}(w)-o_{j,t}}{p_j} \Bigg\rceil_0 e_j,
\end{align}
where $\lceil x \rceil_0$ is lower-bounded by zero.\footnote{Hence, those tasks $\tau_j$ that are released on or after the end of the busy interval, i.e., $W_{h,t}^{k}(w)\le o_{j,t}$, are excluded from the summation in \eqref{eq:busy_period}.} Here, $W_{h,t}^{0}(w)$, i.e., the base busy interval, and $hpe(\tau_h)$ in the summation term are defined differently depending on whether $\tau_h$ is active at $t$: 
\begin{itemize}

\item If $\tau_h$ is active at time $t$ (Figure~\ref{fig:busy_period}): 
	\begin{itemize}
	\item $W_{h,t}^{0}(w)=w + \tilde{e}_{h,t} + \sum_{\tau_j\in hp(\tau_h)} \tilde{e}_{j,t}$, i.e., the total amount of any outstanding workload for jobs of $\tau_h$ and $hp(\tau_h)$ at time $t$ plus the initial window $w$,
	\item $hpe(\tau_h) = hp(\tau_h)$.
	\end{itemize}

\item If $\tau_h$ is not active at time $t$:
	\begin{itemize}
	\item $W_{h,t}^{0}(w)=w + \sum_{\tau_j\in hp(\tau_h)} \tilde{e}_{j,t}$, i.e., the total amount of any outstanding jobs of $hp(\tau_h)$ at $t$ plus the initial window $w$,
	\item $hpe(\tau_h) = hp(\tau_h) \cup \{\tau_h\}.$\footnote{That is, the summation term in \eqref{eq:busy_period} includes $\tau_h$'s upcoming job if the busy interval grows to contain the release of $\tau_h$.}	
	\end{itemize}

\end{itemize}
In the second case, the busy interval may not include the upcoming job of $\tau_h$ (that is released after $t$) if the busy interval ends before its release, i.e., $W_{h, t}^{k}(w) \le o_{h,t}$ for any $k$. Note that \eqref{eq:busy_period} is similar to the computation of dynamic slack \cite{Davis:1993}. Using the iterative procedure, $W_{h, t}(w)$ is computed as follows:
\[\small 
	W_{h, t}(w) =\begin{cases}
	W_{h,t}^{k+1}(w) = W_{h,t}^{k}(w) & {\small \text{if converging for some } k}\\
	\infty &{\small \text{if not converging} }
	\end{cases}.
\]
$W_{h, t}(w)$ is the worst-case (i.e., longest) busy interval that starts with an execution of size $w$ at $t$ by a job of $lp(\tau_h)$. 
Finding $W_{h, t}(w)$ can be viewed as a \emph{simulation} of a priority inversion of size $w$. The scheduler tests if $\tau_h$ would still meet its deadline with the priority inversion in addition to the maximum interference from $hp(\tau_h)$ by checking if the worst-case busy interval ends by the \emph{effective deadline}
\begin{align}\label{eq:schedulability_test}
t + W_{h, t}(w) \le \overline{d}_{h,t}.
\end{align}
Recall that when $\tau_h$ is not active at time $t$, $\overline{d}_{h,t}$ is the deadline of the upcoming job of $\tau_h$ (see Definition~\ref{def:effective_d_h_t}).

\begin{algorithm}[t]
{\small
\caption{\textsc{FindCandidates}($t$, $\Gamma$, $L_{\mathcal{R}_t}$)} \label{alg:find_cand}
\begin{algorithmic}[1]
\STATE $L_\mathcal{C}\gets \tau_{(1)}$ \COMMENT{i.e., the highest-priority job}

\FOR{$\tau_{(i)} = \tau_{(2)}, \ldots, \tau_{(|L_{\mathcal{R}_t}|)}$}
	\FOR{$\tau_{h} \in hp(\tau_{(i)}) $}	
	\STATE $w \gets 1$ \COMMENT{Priority inversion of size 1}
	\STATE $W_{h,t}^{0}(w) \gets w + \sum_{\tau_j\in hp(\tau_h)} \tilde{e}_{j,t}$
	\STATE $hpe(\tau_h) \gets hp(\tau_h)$
	\STATE $\overline{d}_{h,t} \gets r_{h,t} + d_h$	
	\IF {$\tau_h$ is active at $t$}
	\STATE $W_{h,t}^{0}(w) \gets W_{h,t}^{0}(w) + \tilde{e}_{h,t}$
	\ELSE
	\STATE $hpe(\tau_h) \gets hpe(\tau_h) \cup \{ \tau_h \}$
	\STATE $\overline{d}_{h,t} \gets \overline{d}_{h,t}  + d_h$	
	\ENDIF
	\STATE Calculate $W_{h,t}(w)$ using \eqref{eq:busy_period}.
	\IF { $t+W_{h,t}(w) > \overline{d}_{h,t}$}
		\STATE \COMMENT {Potential deadline miss of $\tau_h$. Stop.}
		\RETURN $L_\mathcal{C}$
	\ENDIF
	\ENDFOR
	\STATE \COMMENT {All $hp(\tau_{(i)})$ are schedulable. }
	\STATE $L_\mathcal{C}\gets \tau_{(i)}$ 
\ENDFOR
\RETURN $L_\mathcal{C}$
\end{algorithmic}}
\end{algorithm}

\begin{theorem}
Suppose a set of tasks $\Gamma$ is schedulable with the fixed-priority preemptive scheduling. Then, the schedules randomized by \NewTS{} is still schedulable. 
\end{theorem}
\begin{proof}
The proof is straightforward as the schedulability test by \eqref{eq:busy_period} and \eqref{eq:schedulability_test} is exact; the actual interference from the higher priority tasks cannot be more than what is assumed in \eqref{eq:busy_period}. 
\end{proof}

Algorithm~\ref{alg:find_cand} outlines the candidate selection procedure in \NewTS{} that has been explained so far. As mentioned earlier, a candidacy test is performed on each 
ready job 
in the priority order (\textsc{Line 2}). The highest-priority job in the queue is always a candidate because no priority inversion occurs due to the execution of the job (\textsc{Line 1}). Now, suppose we are testing $\tau_{(i)}$ (\textsc{Lines 3}). For the job to be added to the candidate list, the condition in \eqref{eq:schedulability_test} must be satisfied for all $\tau_h \in hp(\tau_{(i)})$ for $w = 1$ if a scheduling decision is to be made at each time slot or $w=\tilde{e}_{(i),t}$ if $\tau_{(i)}$ is to complete in the priority inversion mode. Algorithm~\ref{alg:find_cand} sets $w=1$ as an example. If there exists at least one $\tau_h \in hp(\tau_{(i)})$ such that \eqref{eq:schedulability_test} does not hold for (\textsc{Line 15}), $\tau_{(i)}$ is not added to the candidate list and the candidate search for time $t$ stops (\textsc{Lines 16--17}). This is because if $\tau_h$ would miss its deadline due to the execution of $\tau_{(i)}$, it would miss the deadline due to $\tau_{(i+1)}$ anyway. The procedure eventually terminates at either \textsc{Line 17} when it finds at least one $\tau_{(i)}$ that may lead some or all of its higher priority tasks to miss deadline, or at \textsc{Line 23} when all jobs in the ready queue are added to the candidate list. 

In fact, not all of $hp(\tau_{(i)})$ need to be examined (\textsc{Line 2}) in Algorithm~\ref{alg:find_cand}. If $\Gamma'_{(i-1)} \subseteq hp(\tau_{(i-1)}) $ is the set that was examined for $\tau_{(i-1)}$ and passed the schedulability test, $\Gamma'_{(i-1)}$ are also schedulable with $\tau_{(i)}$'s priority inversion. This is because the analysis in \eqref{eq:busy_period} depends only on the size of a priority inversion, i.e., $w$, not on who is making the priority inversion. Hence, for $\tau_{(i)}$ we only need to consider $\Gamma'_{(i)}= hp(\tau_{(i)}) - \Gamma'_{(i-1)}$. Therefore, at most $N$ (i.e., the task set size) tests are performed in Algorithm~\ref{alg:find_cand}. Hence, the algorithm can be implemented by a single loop over the task set, examining each task $\tau_h$'s schedulability for a fixed size of priority inversion. 

\begin{example}
Let us consider the example task set used in Example~\ref{example:prob_TS}.
\begin{center}
\includegraphics[width=0.95\columnwidth]{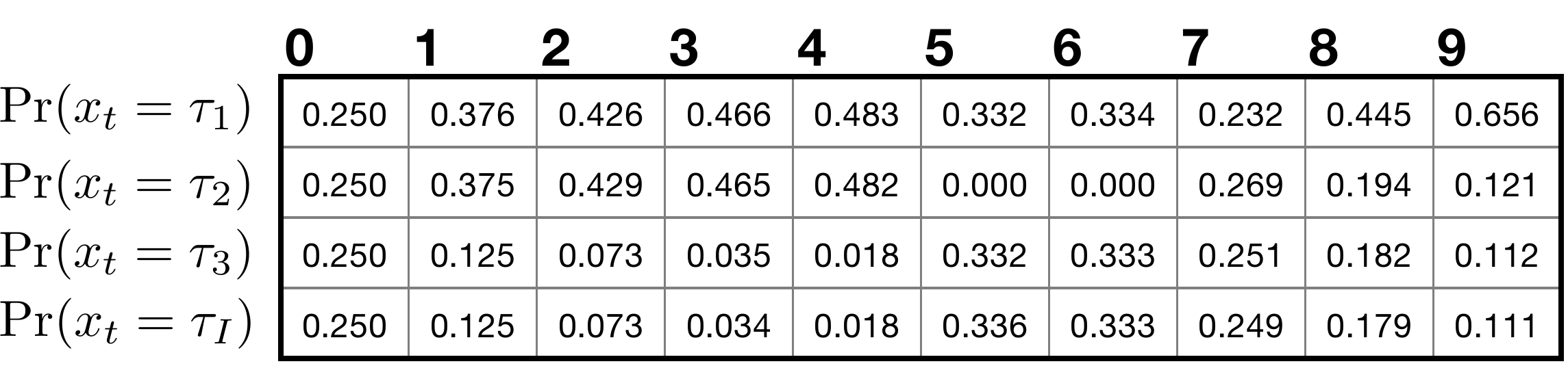}
\end{center}
The table above shows the probability of seeing each task in the first 10 time slots in each hyper-period when randomized by \NewTS{}. We can see that tasks execute over a wider range of time, effectively increasing the uncertainty in task execution at each time slot. In Section~\ref{sec:entropy}, we discuss what this result means in terms of the worst-case vulnerability against timing inference attack. 
\end{example}

\subsection{Approximate \NewTS{}}\label{subsec:approx_newTS}

In this section we present an approximate \NewTS{} algorithm that does not require the exact simulation of priority inversion during the candidate selection. Suppose a candidate selection is being made at time $t$. As in the exact algorithm, for each job $\tau_{(i)} \in L_{\ReadySet}$, starting from the highest priority, the scheduler checks if $\tau_{(i)}$'s priority inversion would still allow to meet deadlines for all of the higher priority tasks $hp(\tau_{(i)})$. Let $\tau_h \in hp(\tau_{(i)})$ be the task being tested. The approximate algorithm considers the two cases that depend on $\tau_h$'s status (active vs. inactive) at time $t$ separately. Let us call them \textsc{Test I} and \textsc{Test A} if $\tau_h$ is inactive and active, respectively, at time $t$.

\vspace{0.5\baselineskip}
\noindent\textbf{1) Test I: Inactive $\tau_h$ at time $t$}

The main difficulty in the approximate \NewTS{} is how to test whether an inactive $\tau_h$ would satisfy its deadline after it is released. For this, \textsc{Test I} is composed of \textsc{Test I-1} and \textsc{Test I-2}. \textsc{Test I-1} tests if a busy interval that begins with a priority inversion at present would end before $\tau_h$'s earliest next arrival. Now, let $o_{h,t}$ be the relative offset of the $\tau_h$'s earliest next arrival from time $t$, which is computed by $o_{h,t} = r_{h,t} + p_h - t$. If the following inequality holds, 
$\tau_h$ that releases at $t + o_{h,t}$ is guaranteed to be schedulable even if a priority inversion of size $w$ by any $lp(\tau_h)$ occurs at $t$:
\begin{align}\label{eq:theorem_D1}
w + \sum_{\tau_{j} \in hp(\tau_h)} \tilde{e}_{j,t} +  \sum_{\tau_{j} \in hp(\tau_h)} \bigg\lceil\frac{o_{h,t}-o_{j,t}}{p_j}\bigg\rceil_0 e_j \le o_{h,t}.
\end{align}

\begin{figure}[t]
\centering
\includegraphics[width=1\columnwidth]{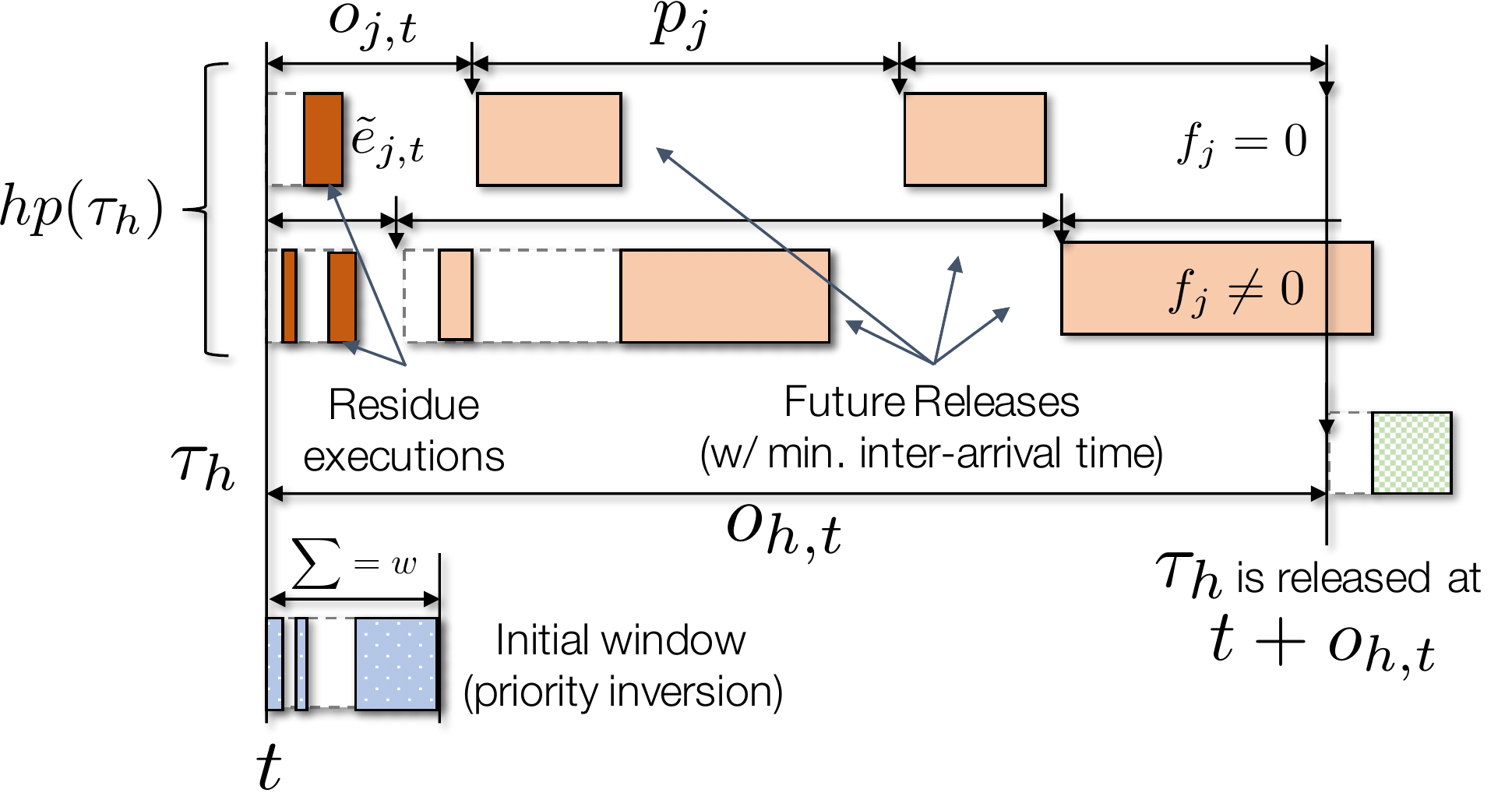}
\caption{Busy interval analysis for Theorem~\ref{theorem:test_I_1}.}\label{fig:theorem_test_I_1}
\end{figure}

\begin{theorem}\label{theorem:test_I_1}
A level-$\tau_h$ busy interval of initial size $w$ at time $t$ ends before or on $\tau_h$'s earliest next arrival if \eqref{eq:theorem_D1} holds. 
\end{theorem}
\begin{proof}	
\eqref{eq:theorem_D1} excludes those tasks that will never arrive by the release of $\tau_h$ due to $\lceil \cdot \rceil_0$. (In what follows, we drop the subscript $0$ from the ceiling and floor functions for the brevity.) For each $\tau_j\in hp(\tau_h)$, it is released at least $ \big \lfloor {(o_{h,t} - o_{j,t})}/{p_j} \big \rfloor$ times. Let $o_{h,j,t} = o_{h,t} - o_{j,t}$ be the distance of $\tau_h$'s release from $\tau_j$'s release. If ${o_{h,j,t}}/{p_j} \notin \mathbb{Z}$, an additional instance of $\tau_j$ is released before $t+o_{h,t}$. It can execute for $f_j$ time units, where $0\le f_j\le e_j$. 
If ${o_{h,j,t}}/{p_j} \in \mathbb{Z}$,
$f_j=0$. Figure~\ref{fig:theorem_test_I_1} shows both cases. 
Now, the time window $[t, t+o_{h,t}]$ is the sum of 
\begin{align}\label{eq:theorem_D2}
o_{h,t} = w + \hspace{-0.1cm} \sum_{\tau_{j} \in hp(\tau_h)}  \hspace{-0.1cm}\tilde{e}_{j,t} +  \hspace{-0.1cm}\sum_{\tau_{j} \in hp(\tau_h)}  \hspace{-0.1cm} \bigg\lfloor\frac{o_{h,j,t}}{p_j}\bigg\rfloor e_j +  \hspace{-0.1cm} \sum_{\tau_{j} \in hp(\tau_h)}\hspace{-0.1cm} f_j + R,
\end{align}
where $R$ is the sum of any idle times that occurs in $[t, t+o_{h,t}]$. We prove the claim by showing that $R$ is non-negative (i.e., the busy period is discontinued) and hence $o_{h,t}$ correctly upper-bounds the left-hand side of \eqref{eq:theorem_D1}. 
First, substituting $o_{h,t}$ in \eqref{eq:theorem_D1} with \eqref{eq:theorem_D2}, 
\begin{align}
\sum_{\tau_{j} \in hp(\tau_h)} \hspace{-0.1cm}\bigg\lceil\frac{o_{h,j,t}}{p_j}\bigg\rceil e_j  \le \hspace{-0.1cm} \sum_{\tau_{j} \in hp(\tau_h)}  \hspace{-0.1cm}\bigg\lfloor\frac{o_{h,j,t}}{p_j}\bigg\rfloor e_j + \hspace{-0.1cm}\sum_{\tau_{j} \in hp(\tau_h)} \hspace{-0.1cm}f_j + R\nonumber
\end{align}
which can be rewritten as
\begin{align}
\sum_{\tau_{j} \in hp(\tau_h)} \bigg( \bigg\lceil\frac{o_{h,j,t}}{p_j}\bigg\rceil  - \bigg\lfloor\frac{o_{h,j,t}}{p_j}\bigg\rfloor \bigg) e_j  - \sum_{\tau_{j} \in hp(\tau_h)} f_j \le R.\nonumber
\end{align}
Because $\lceil x \rceil - \lfloor x \rfloor$ = 0 or 1 if $x\in \mathbb{Z}$ or $x\notin \mathbb{Z}$, respectively, 

\begin{align}
\sum_{\substack{\tau_{j} \in hp(\tau_h),  {o_{h,j,t}}/{p_j} \notin \mathbb{Z}}} (e_j - f_j) \le R. \nonumber
\end{align}
Since $0\le f_j \le e_j$, the left hand-side is non-negative, so is $R$. Since $R\ge 0$, \eqref{eq:theorem_D2} leads to \eqref{eq:theorem_D1}.
Therefore, the busy interval of initial size $w$ that begins at $t$ cannot be longer than $o_{h,t}$ if \eqref{eq:theorem_D1} holds.  
\end{proof}

In other words, if \eqref{eq:theorem_D1} is satisfied, any job in $lp(\tau_h)$ is allowed to execute for $w$ while not imposing any additional delay on the upcoming $\tau_h$ due to the priority inversion because the busy interval will finish by $\tau_h$'s release at $t+o_{h,t}$.

However, \eqref{eq:theorem_D1} is not a necessary condition for $\tau_h$'s schedulability. That is, if \eqref{eq:theorem_D1} does not hold, $\tau_h$ may or may not experience a delay due to deferred executions of $hp(\tau_h)$. To avoid the exact test, we perform \textsc{Test I-2}, an approximate test with a pessimistic assumption about the jobs of $hp(\tau_h)$ that are released before $\tau_h$'s release. In \textsc{Test I-2}, we first compute the worst-case residue executions of $hp(\tau_h)$ at $\tau_h$'s upcoming release (at time $t' = t+ o_{h,t}$) and treat all of them as deferred executions. For the brevity in discussion, let $\rho_{h,t'}$ denote the sum of the residues and call it the \emph{overflow} of $hp(\tau_h)$ at $t'$: 
\begin{align}
\rho_{h,t'} = \sum_{\tau_j \in hp(\tau_h)} \tilde{e}_{j,t'}.\nonumber 
\end{align}
Finding the correct value of $\rho_{h,t'}$ requires an iterative procedure as $\tilde{e}_{j,t'}$ is unknown at $t$. Hence, we find an upper bound $\overline{\rho}_{h,t,t'}$ instead as shown in Figure~\ref{fig:theorem_test_I_2}:
\begin{align}
\overline{\rho}_{h,t,t'} = \sum_{\tau_j \in hp(\tau_h)} \hspace{-0.2cm}\bigg (\alpha_{t,t',j} \cdot  e_j + (1-\alpha_{t,t',j})\cdot \tilde{e}_{j,t} \bigg)  -  (t' - r^*),\nonumber
\end{align}
where 
\begin{itemize}
\item $\alpha_{t,t',j}=1$ if $\tau_j$ arrives in $[t, t')$ with their minimum inter-arrival times, i.e., if $o_{j,t} < t'-t$. Otherwise, 0. 
\item $r^*$ is the latest release time among those $\tau_j \in hp(\tau_h)$ that release in $[t, t')$, and computed by 
\[r^* = \max_{\substack{\tau_j\in hp(\tau_h)\\  \alpha_{t,t',j}=1}} \bigg(o_{j,t} + \bigg\lfloor \frac{(t'-t)-o_{j,t}}{p_j} \bigg\rfloor p_j \bigg).\]

\end{itemize}

\begin{figure}[t]
\centering
\includegraphics[width=1\columnwidth]{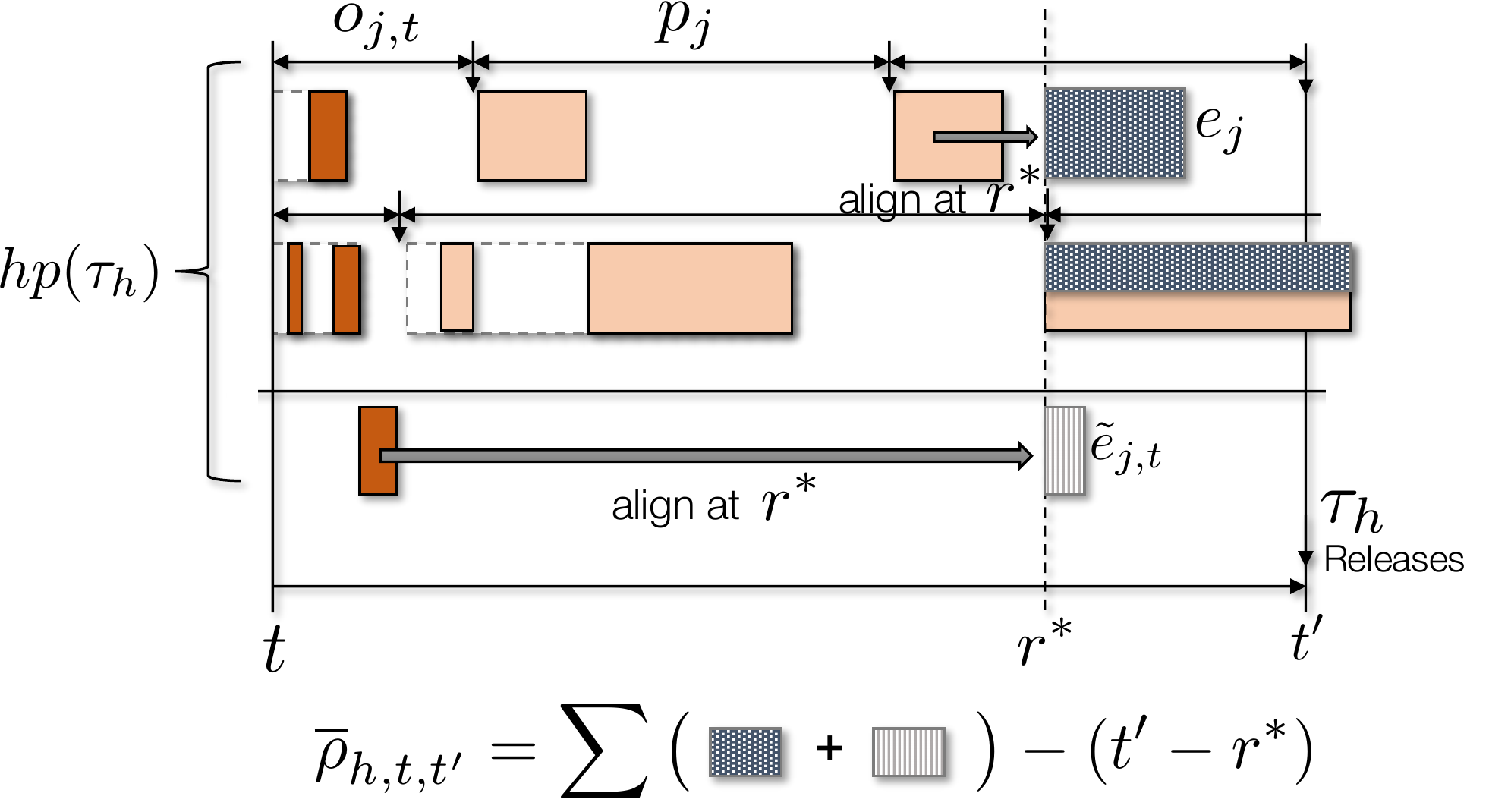}
\caption{Worst-case bound on the overflow of $hp(\tau_h)$ at $\tau_h$'s release at $t'$.}\label{fig:theorem_test_I_2}
\end{figure}

That is, $\overline{\rho}_{h,t,t'}$ is obtained by assuming as if each instance of $\tau_j \in hp(\tau_h)$ is delayed until the 
latest release point among them. This approximate residue at time $t'$ when $\tau_h$ releases is an upper bound of the true residue $\rho_{h,t'}$. 
\begin{lemma}\label{lemma:rho_upperbound}
$\overline{\rho}_{h,t,t'}$ is an upper bound of $\rho_{h,t'}$.
\end{lemma}
\begin{proof}
Let us first consider those $\tau_j$ that have job release(s) in the window $[t, t')$, i.e., $\alpha_{t,t',j}=1$. Let $E_1$ denote the sum of their last executions. Now, let $E_2$ be the sum of $\tilde{e}_{j,t}$ of those $\tau_j$ that have no job release in the window, i.e., $\alpha_{t,t',j}=0$. Let $E$ be the sum of their remaining executions at time $r^*$. Then, $E=E_1+E_2-\epsilon$ for a non-negative $\epsilon$ because of time progress. Since there is no more new job invocation from $r^*$, the total residue at $t'$ is
\[\sum_{\tau_j \in hp(\tau_h)} \tilde{e}_{j,t'} = E_1+E_2-\epsilon-(t'-r^*)\le E_1 + E_2 - (t' - r^*),\]
since $\epsilon$ is non-negative. Therefore, $\rho_{h,t'} \le \overline{\rho}_{h,t,t'}$.  
\end{proof}
Computing $\overline{\rho}_{h,t,t'}$ is the crux of \textsc{Test I-2}. It treats $\overline{\rho}_{h,t,t'}$ as the maximum amount of deferred executions by $hp(\tau_h)$ at time $t'$ (i.e., released before but not finished by $t'$). Then, the scheduler tests if $\tau_h$ would still meet its deadline even in the presence of such worst-case delay. 
For this, we need to define the \emph{maximum slack} of $\tau_h$:
\begin{definition}[$\overline{V}_h$: Maximum Slack of Task $h$]\label{def:max_slack}
$\overline{V}_h$ is the maximum amount of time that $\tau_h$ can additionally have while meeting its deadline 
when there are \emph{no deferred executions} of $hp(\tau_h)$ at $\tau_h$'s release. That is,
\begin{align}\label{eq:nominal_budget}
\overline{V}_h = \operatorname*{argmax}_{q} \Big( e_h + q + \mathbf{I_h} (e_h + q) \Big ) \le d_h,
\end{align}
where $\mathbf{I_h}(e_h + q)$ is a function that returns the worst-case interference from $hp(\tau_h)$ that $\tau_h$ can experience when its execution time is bloated to $e_h+q$ and $\tau_h$ is released together with all of $hp(\tau_h)$.
\end{definition}
The maximum slack of $\tau_h$ is calculated off-line using the response time analysis in \eqref{eq:wcrt_i} by replacing $e_h$ with $e_h + q$ and increasing $q$ from 0 while the worst-case response time of $\tau_h$ does not exceed the deadline $d_h$.
Note that by the definition and \eqref{eq:wcrt_i}, the maximum slack of every task $\tau_h\in \Gamma$ is non-negative for a schedulable task set $\Gamma$. 
Also note that $\overline{V}_h$ is always greater than or equal to $V_h$, i.e., the worst-case inversion budget used in \TS{} (Section~\ref{sec:TaskShuffler}).

\begin{example}
Let us consider the example task set used in Example~\ref{example:prob_TS} again. Then, $\overline{V}_1 = 3$, $\overline{V}_2 = 1$, and $\overline{V}_3 = 3$. This means that, for instance, $\tau_2$ can have additional 1 time unit while not missing its deadline if there was \emph{no deferred execution} of $\tau_1$ when $\tau_2$ is being released. 
\end{example}
It is important to note that $\overline{V}_h$ is the budget when there are \emph{no} deferred executions by $hp(\tau_h)$ when $\tau_h$ is being released. Hence, if there \emph{are} deferred executions and the sum of them exceeds $\overline{V}_h$, $\tau_h$ may miss its deadline. Therefore, \textsc{Test I-2} checks if $\overline{\rho}_{h,t,t'}>\overline{V}_h$. If true, no priority inversion is allowed by any of $lp(\tau_h)$ $t$. 
Theorem~\ref{theorem:v_h_schedulable} will prove that this ensures that $\tau_h$ will not miss its deadline.

\vspace{0.5\baselineskip}
\noindent\textbf{2) Test A: Active $\tau_h$ at time $t$}

This test is simpler than \textsc{Test I}. \textsc{Test A} simply uses a counter called \emph{remaining inversion budget}, $v_h$, for each task $\tau_h$. It represents the amount of time that can be given to the lower priority tasks $lp(\tau_h)$ to execute in a priority inversion mode while $\tau_h$ is active. $v_h$ is initialized at each release of $\tau_h$, but varies depending on the run-time state of $hp(\tau_h)$. In the simplest case, it could be set to $\overline{V}_h$ less the amount of deferred executions by $hp(\tau_h)$ at $\tau_h$'s release. Then, by the definition of $\overline{V}_h$, $\tau_h$ is guaranteed to meet its deadline. However, finding the correct amount of the deferred executions is complex as this requires backtracking on schedule. Thus, we simply treat any residue executions at $\tau_h$'s release as deferred executions. 
However, this leads to a pessimistic bound on $v_h$ because (i) not all (or even none) of the residue executions need be deferred executions and (ii) $\overline{V}_h$ is obtained with the assumption that all of $hp(\tau_h)$ are released together with $\tau_h$. In fact, the scheduler already has all the information needed to calculate a less-pessimistic budget $v_h$. This includes (i) the most recent release times of $hp(\tau_h)$ and (ii) their residue executions at $\tau_h$'s release.

Suppose $\tau_h$ is being released at time $t$. Then we calculate the upper bound on the interference from $hp(\tau_h)$ during the window $[t, t+d_h]$, i.e., until $\tau_h$'s deadline. 
Denoting it by $I_{h, t}$, 
\begin{align}\label{eq:I_h_t}
I_{h,t} = \sum_{\tau_{j} \in hp(\tau_h)} \bigg( \tilde{e}_{j,t} + \bigg\lfloor\frac{d_h-o_{j,t}}{p_j}\bigg\rfloor e_j + f_{j,t} \bigg),
\end{align}
where $f_{j,t}$ is the maximum amount of execution by the last release of $\tau_j$ by the end of the window and calculated by
\[f_{j,t} = \min \bigg (e_j, \quad d_h - \Big(o_{j,t} + \Big\lfloor\frac{d_h-o_{j,t}}{p_j}\Big\rfloor p_j \Big)  \bigg). \]
Note that the residue execution time $\tilde{e}_{j,t}$ and the earliest next arrival time $o_{j,t}$ of higher-priority task $\tau_{j}$ are already known at $t$. Then, $v_h$ is initialized to
\begin{align}\label{eq:v_h}
v_h = d_h - e_h - I_{h,t}.
\end{align}

\begin{theorem}\label{theorem:v_h_schedulable}
Each task $\tau_h$ is schedulable as long as it allows $lp(\tau_h)$ to execute for at most $v_h$ time units while $\tau_h$ is active. 
\end{theorem}
\begin{proof}
This is equivalent to proving that $v_h$ in \eqref{eq:v_h} is non-negative, i.e., $e_h + I_{h,t}\le d_h$. First of all, $I_{h,t}$ in \eqref{eq:I_h_t} can be rewritten as
\begin{align}\label{eq:another_I_h_t}
I_{h,t} = \hspace{-0.2cm}\sum_{\tau_{j} \in hp(\tau_h)} \tilde{e}_{j,t}  +  \hspace{-0.2cm}\sum_{\tau_{j} \in hp(\tau_h)} \bigg( \bigg\lfloor\frac{d_h-o_{j,t}}{p_j}\bigg\rfloor e_j + f_{j,t} \bigg). 
\end{align}
The first summation is $\rho_{h,t}$, i.e., the overflow of $hp(\tau_h)$ at time $t$. Recall that the algorithm, in particular \textsc{Test I-2}, ensures that the upper bound of the overflow (thus the actual overflow, $\rho_{h,t}$, itself as well) does not exceed $\tau_h$'s maximum slack, $\overline{V}_h$, as explained previously. Now, the sum of $\tau_h$'s execution and the overflow can be viewed as an execution of size $e_h + \rho_{h,t}$. This busy interval cannot grow longer than $d_h$ due to the definition of $\overline{V}_h$. 
That is, by \eqref{eq:nominal_budget} and $\rho_{h,t} \le \overline{V}_h$,  
\[e_h + \rho_{h,t} + \mathbf{I_h}(e_h + \rho_{h,t}) \le e_h + \overline{V}_h + \mathbf{I_h}(e_h + \overline{V}_h) \le d_h.\]
Here $\mathbf{I_h}(\cdot)$, as previously explained, assumes the worst-case release pattern of $hp(\tau_h)$: when they are aligned with $\tau_h$'s release. Hence, the second summation in \eqref{eq:another_I_h_t} is upper-bounded by $\mathbf{I_h}(e_h + \rho_{h,t})$ above. Therefore, 
\[e_h + I_{h,t} \le e_h + \rho_{h,t} + \mathbf{I_h}(e_h + \rho_{h,t}) \le d_h,\]
which proves the theorem.
\end{proof}

\vspace{0.3\baselineskip}
\noindent\textbf{Summary of Approximate \NewTS{}:} The approximate \NewTS{} has a similar algorithmic structure as the exact version in Section~\ref{subsec:exact_newTS}, except that the maximum slack $\overline{V}_i$ of every $\tau_i\in \Gamma$ is calculated off-line first in the approximate version. Then, for each ready job $\tau_{(i)}$, a candidacy test is performed by testing each $\tau_h \in hp(\tau_{(i)})$ against \textsc{Test I} or \textsc{Test A} depending on $\tau_h$'s states (active vs. inactive) at the time of the candidacy test. If $\tau_h$ is inactive, \textsc{Test I-1} is first applied. If it fails, \textsc{Test I-2} is performed. If it fails again, $\tau_{(i)}$ is not added to the candidate list and the candidate search stops, similarly to the exact algorithm. When $\tau_h$ is released, its inversion budget $v_h$ is initialized to \eqref{eq:v_h}. From then, the test of $\tau_h$ is done by \textsc{Test A}, i.e., checking if $v_h \ge 0$. As in the case of the exact algorithm, this candidate selection process can be implemented by a single loop over the task set. 

The critical path of the approximate \NewTS{} is when almost every task is inactive. For each inactive task $h$, we iterate over $hp(\tau_{h})$ to calculate the projected interference. Hence, worst-case time complexity is $O(N^2)$.

\subsection{Weighted Job Selection}\label{subsec:weighted_selection}

\begin{figure}[t]
\centering
\includegraphics[width=1\columnwidth]{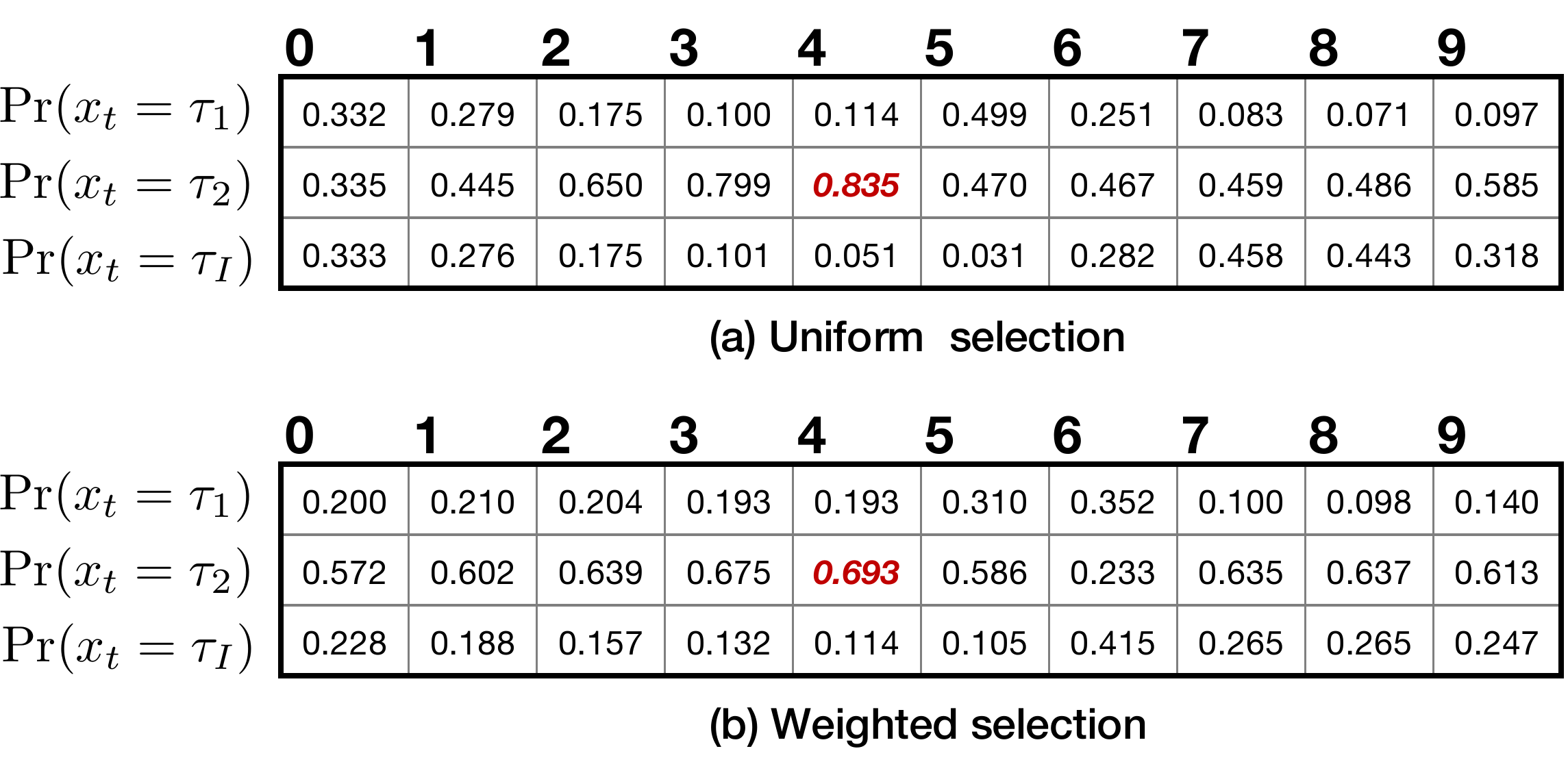}
\caption{Probability of task execution at each time slot for two tasks $\tau_1:=(p_1=5, e_1=1)$ and $\tau_2:=(p_2=7,e_2=4)$.
}
\label{fig:weighted_job_selection}
\end{figure}

Now, given a list of candidate jobs, the scheduler picks one randomly. In \TS{} \cite{YoonTaskShuffler}, the selection is done with the uniform probability: each job has an equal chance of $\frac{1}{n}$ where $n$ is the number of candidate jobs. Counter-intuitively, this leads some tasks to appear in particular time slots more often. Consider an example in Figure~\ref{fig:weighted_job_selection}(a), which shows the probability of seeing a particular task at each time slot for the first 10 time slots in each hyper-period when two tasks $\tau_1:=(p_1=5, e_1=1)$ and $\tau_2:=(p_2=7,e_2=4)$ run under \NewTS{} for 100,000 hyper-periods (i.e., 3.5 million time slots). The probability of seeing $\tau_2$ at time 4 in each hyper-period is $83.5\%$. Such a high probability is because $\tau_1$, which has a smaller execution, would be \emph{unlikely} to have a remaining execution as time proceeds. This is in turn due to the equal chance of being selected at each slot. Accordingly, it is highly probable that $\tau_2$ would have a remaining execution at time $4$. 

In order to alleviate such biases, we propose a \emph{weighted} selection process that considers the remaining execution until the deadline. Suppose a selection is made at time $t$ from the candidate list $L_{\mathcal{C}}$. For each task $\tau_i$ in the list, the scheduler uses the \emph{remaining utilization} 
\[u_{i,t} = \tilde{e}_{i,t}/({d_{i,t}-t}),\]
where $\tilde{e}_{i,t}$ and $d_{i,t}$ are the residue execution and the current deadline of $\tau_i$ at time $t$ (defined in Section~\ref{subsec:exact_newTS}), respectively. 
Then, each task is assigned a normalized weight  
$\omega_{i,t} = u_{i,t}\big/{\sum_{\tau_k \in L_{\mathcal{C}}}u_{k,t}}$.  The scheduler performs a weighted random selection based on the $\omega_{i,t}$ values. Note that $u_{i,t}$ directly reflects the probability of the remaining execution to appear until the deadline. Hence, it is desired to keep $\max_i[u_{i,t}]$ small for any time $t$ to reduce the chance of correct timing inference. If $u_{j,t}>u_{k,t}$ for two jobs $\tau_j$ and $\tau_k$ at time $t$, selecting $\tau_j$ may lead to $\max_i[u_{i,t}]>\max_i[u_{i,t+1}]$ while selecting $\tau_k$ always leads to $\max_i[u_{i,t}]\le \max_i[u_{i,t+1}]$ which indicates increased chance of correct timing inference (discussed shortly in Section~\ref{subsec:slot_min_entropy}). Therefore, giving a higher chance to a job with a higher remaining utilization (i.e., $\tau_j$ if $u_{j,t}>u_{k,t}$) reduces the chance of temporal locality. 
Figure~\ref{fig:weighted_job_selection}(b) shows the result of applying this process to the two-task example used above. It should be noted that the goal of the weighted selection is not to make the probabilities of different tasks as equal as possible, but to make the probability of seeing a task 
as invariant over time as possible. 
In the next section, we discuss the implication of this reduced bias from an adversary's perspective.

\section{Schedule Min-Entropy}
\label{sec:entropy}

In \cite{YoonTaskShuffler}, \emph{Schedule Entropy} is introduced as a measure of the amount of uncertainty in a randomized schedule of a task set. It is defined over the probability distribution of the hyper-period schedules. 
However, it cannot quantify how difficult it is to make a correct prediction on task execution at an \emph{arbitrary} time instant. In this section, we introduce a new metric that measures the \emph{worst-case vulnerability} of an arbitrary schedule to timing inference. 

\subsection{Slot Min-Entropy}\label{subsec:slot_min_entropy}

Let us consider Figure~\ref{fig:weighted_job_selection}(a) again. $x_t$ is a random variable, whose domain is $\chi = \Gamma \cup \{ \tau_I\}$, indicating which task executes at time slot $t$. 
In the example, the Shannon entropy~\cite{Shannon:1948} over $\Pr(x_2)$, which is about 1.29, is almost equal to that over $\Pr(x_8)$. However, when an adversary is to make a guess of which task would execute, he/she has a higher chance of making a correct guess in the \emph{first} attempt for $t=2$ than for $t=8$ because $\max_{x\in \chi}(\Pr(x_2=\tau_x)) = 0.650$ is larger than $\max_{x\in \chi}(\Pr(x_8=\tau_x)) = 0.486$. This is similar to the notion of \emph{vulnerability} \cite{Smith:2009} in quantitative information flow analysis. In our context, $\max_{x\in \chi}(\Pr(x_t=\tau_x))$ is the worst-case probability, or the best-case probability from the adversary's perspective, that the adversary can make a correct prediction on task execution at time $t$ in a single attempt. If it is 1, the adversary can always correctly determine which task will run at $t$. Hence, a low value of $\max_{x\in \chi}(\Pr(x_t=\tau_x))$ is desired to reduce the possibility of correct prediction.

Often the vulnerability is expressed in entropy called \emph{min-entropy} \cite{renyi1961, Smith:2009}. Applying it to the context of timing inference attack, we define \emph{slot min-entropy}:
\begin{definition} 
The slot min-entropy of time slot $t$ is 
\begin{align}
H_{\infty}(x_t)=-\log\Big( \max_{\tau_x\in \Gamma} \big(\Pr(x_t=\tau_x) \big) \Big).
\end{align}
For instance, $H_{\infty}(x_2)=0.431$ and $H_{\infty}(x_8)=0.722$ for the example in Figure~\ref{fig:weighted_job_selection}(a). 
\end{definition}
Notice that the domain of $x_t$ here does not include the idle task. That is, we do not concern the adversary's correct guess on the idle task's execution. The slot min-entropy is lower-bounded by $0$ and is upper-bounded by the Shannon entropy. What a slot min-entropy $H_{\infty}(x_t)$ tells is that the worst-case probability of correct guess is at best $(1/2)^{H_{\infty}(x_t)}$. 

\subsection{Schedule Min-Entropy}\label{subsec:sched_min_entropy}
As a hyper-period schedule-level measure, we define the \emph{schedule min-entropy} as follows.
\begin{definition} 
The schedule min-entropy of schedule $S$, where $S$ is the set of time slots over the hyper-period of length $L$, i.e., $S=\{x_0, x_1, \ldots, x_{L-1}\}$,
is the minimum of the slot min-entropies:
\begin{align}
H_{\infty}(S) = \min_{x_t \in S} H_{\infty}(x_t).
\end{align}
\end{definition}
\revised{Accordingly, $(1/2)^{H_{\infty}(S)}$ is the best-case probability that an adversary can correctly guess, \emph{by chance}, which task would run at an \emph{arbitrary} time slot over $S$. In the example in Figure~\ref{fig:weighted_job_selection}(a), $H_{\infty}(S) = 0.206$ (at $t=18$, which is not shown). It is $0.422$ at $t=19$ in the example in (b). These indicate that if the adversary makes a random guess of $x_t = \tau_2$ for every slot $t$, the chance of correct guess is $86.7\%$ and $74.6\%$, respectively. These are the worst-case from the defender's perspective, and hence ${H_{\infty}(S)}$ captures the `weakest link' among all tasks for the whole observation duration. The defender would like to decrease this worst-case chance (i.e., the adversary's best-chance). If $H_{\infty}(S_1)<H_{\infty}(S_2)$ for two schedules $S_1$ and $S_2$, $S_2$ can be said more secure than $S_1$ against the adversary's best guess. Note that the schedule min-entropy does not quantify the information embedded in the order of task executions as $x_t$ is not conditioned on $x_0, \ldots, x_{t-1}$. If the attacker is able to observe the schedule up to present, i.e., $x_0, \ldots, x_{t-1}$, as is the case in \cite{Nasri2019}, a conditional entropy can model the vulnerability due to the attacker's observation on the execution order.}

The following theorem establishes an upper bound on the schedule min-entropy of a given task set.\footnote{The lower bound is trivially zero, which is obtained when schedule is fixed. }
\begin{theorem}\label{theorem:upper_sched_min_entropy}
The upper bound on $H_{\infty}(S)$ of a task set $\Gamma$ is
\begin{align}
\overline{H}_{\infty}(S) =  -\log\Big( \max_{\tau_x\in \Gamma} u_x \Big).
\end{align}
\end{theorem}
\begin{proof}
    For any schedulable task $\tau_x$, it executes for $e_x$ over its period $p_x$. Hence, if the execution is evenly distributed (i.e., ideal randomization), $\Pr(x_t=\tau_x)=\frac{e_x}{p_x}$ for an arbitrary time slot $t$, where $\frac{e_x}{p_x}=u_x$ is the utilization of $\tau_x$ as explained in Section~\ref{subsec:system_model}. If the execution is not evenly distributed, there exists at least one slot $t$ for which $\Pr(x_t=\tau_x)\ge u_x$ since the sum of $\Pr(x_t=\tau_x)$ over the period is $e_x$. This means that 
    \begin{align}
    \max_{x_t\in S} \Pr(x_t=\tau_x)\ge u_x.\label{eq:max_pr}
    \end{align}
    Now, by the definitions of $H_{\infty}(S)$ and $H_{\infty}(x_t)$ and \eqref{eq:max_pr}
    \begin{align}
    H_{\infty}(S) 
    &\hspace{-0.1cm} = \hspace{-0.1cm}-\log \bigg[ \max_{\tau_x\in \Gamma}   \Big(\max_{x_t \in S}  \big(\Pr(x_t=\tau_x) \big)   \Big) \bigg]\hspace{-0.1cm} \le \hspace{-0.1cm}-\log \Big( \max_{\tau_x\in \Gamma}   u_x    \Big). 
    \nonumber
    \end{align}
\end{proof}
\vspace{-0.3cm}
Hence, in the adversary's best-case (thus the defender's worst-case), the probability of making a correct guess on task execution at an arbitrary time cannot be lower than $(1/2)^{\overline{H}_{\infty}(S)} = \max_{\tau_x\in \Gamma}u_x$, i.e., the largest task utilization. Therefore, a task set of smaller workload tends to have a higher schedule min-entropy (hence less vulnerable to timing inference), as we will see in the next section.

\section{Evaluation}
\label{sec:eval}

\subsection{Evaluation Setup}
\label{subsec:eval_setup}

We use the same parameters as in \cite{YoonTaskShuffler} to generate random synthetic task sets. 
Total $6000$ sets are evenly generated from ten base
utilization groups, $[0.02+0.1\cdot i, 0.08+0.1\cdot i]$ for $i = 0, . . . ,
9$. The base utilization of a set is defined as the total sum of the task utilizations. 
 Each group has six sub-groups, each of which has a fixed number of tasks -- $5$, $7$, $9$, $11$, $13$ and $15$. 
 This is to
generate task sets with an even distribution of tasks. Each task period
is a divisor of $3000$ (but not smaller than $10$). This is to set {\em a common
hyper-period} (i.e., $3000$) over all the task sets. The task execution times are
randomly drawn from $[1,50]$. The deadline for each task is the same as
its period and priorities are assigned according to the Rate Monotonic
algorithm~\cite{LiuLayland1973}. As mentioned earlier, we avoid introducing release jitter in this paper in order to isolate its impacts on the randomness and also because jitter effects were evaluated in~\cite{YoonTaskShuffler} (Figure 12). 

All of the $6000$ random sets are guaranteed to be \emph{schedulable} by the fixed-priority preemptive scheduling~\cite{Audsley93}. For each task set, we run the simulation for 100,000 hyper-periods. Tasks execute for their worst-case execution times because this creates the worst-case situation for the defender -- \revised{that is, removing the randomnesses only makes the timing inference easier for the adversary. Later, we also evaluate the impact of varying execution times (Figure~\ref{fig:exp_varying_exec_time}).}

In what follows, we compare the following three methods:
\begin{itemize}
\item \textsc{TS}: \TS{} algorithm in \cite{YoonTaskShuffler}, 
\item \textsc{TS++ Exact}: The exact \NewTS{} algorithm presented in Section~\ref{subsec:exact_newTS}, and
\item \textsc{TS++ Approx}: The approximate \NewTS{} algorithm presented in Section~\ref{subsec:approx_newTS}.
\end{itemize}
Unless otherwise specified, both \textsc{TS++} use the weighted job selection. 

\subsection{Results}
\label{subsec:results}

\begin{figure}[b]
  \centering
  \includegraphics[width=1\columnwidth]{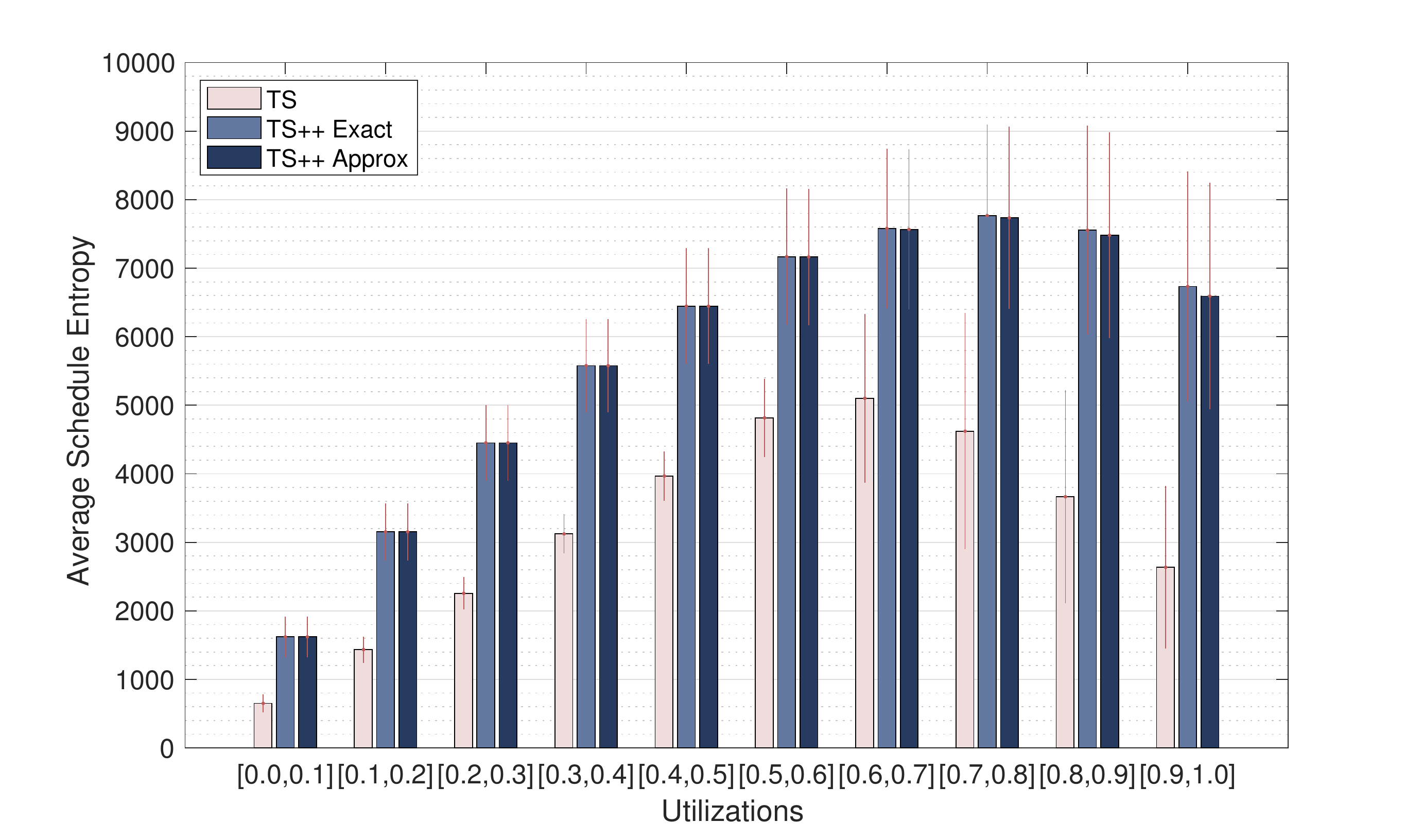}
  \caption{The average schedule entropy with \textsc{TS} and \textsc{TS++}. }
  \label{fig:exp_old_type_entropy}
  \end{figure}

We first compare the schedule entropy, which was used in \cite{YoonTaskShuffler}, of the three methods. Figure~\ref{fig:exp_old_type_entropy} shows the average schedule entropy for each utilization group. The \textsc{TS++} methods (both \textsc{Exact} and \textsc{Approx}) significantly increase the schedule entropy by removing the pessimism in priority inversion budgets in \textsc{TS}. We can also observe that the schedule entropy of \textsc{TS} decreases when the utilization is high. As discussed in \cite{YoonTaskShuffler} in detail, this is mainly because higher priority tasks have less budgets for priority inversion when they have high utilization. The level-$\tau_x$ exclusion policy of \textsc{TS} further decreases uncertainties due to the restrictive candidate selection. On the other hand, \textsc{TS++} methods do not suffer such problems. While a decreased randomness is unavoidable when the system is highly packed, the reduction in \textsc{TS++} methods is small due to the absence of the static (thus pessimistic) bounds on inversion budgets.

\begin{table*}[t]
  \caption{The percentage of task sets that have zero schedule min-entropy. }
  \label{table::exp_num_of_zero_entropy}\centering
  {\footnotesize
  \begin{tabular}{|c||c|c|c|c|c|c|}
  \hline
    Utilizations  & [0.4, 0.5] & [0.5, 0.6] & [0.6, 0.7] & [0.7, 0.8] & [0.8, 0.9] & [0.9, 1.0]\\ \hline\hline
    \textsc{TS}  & 0.50\% & 5.33\% & 17.50\% & 40.67\% & 69.33\% & 92.33\% \\\hline    
    \textsc{TS++ Approx}  & 0.00\%  & 0.00\% & 0.67\% & 3.50\% & 12.00\% & 28.67\% \\\hline            
    \textsc{TS++ Exact}  & 0.00\% & 0.00\% & 0.00\% & 0.00\% & 0.00\% & 0.00\% \\\hline
  \end{tabular}
  }
  \end{table*}

\begin{figure}[t]
\centering
\includegraphics[width=1\columnwidth]{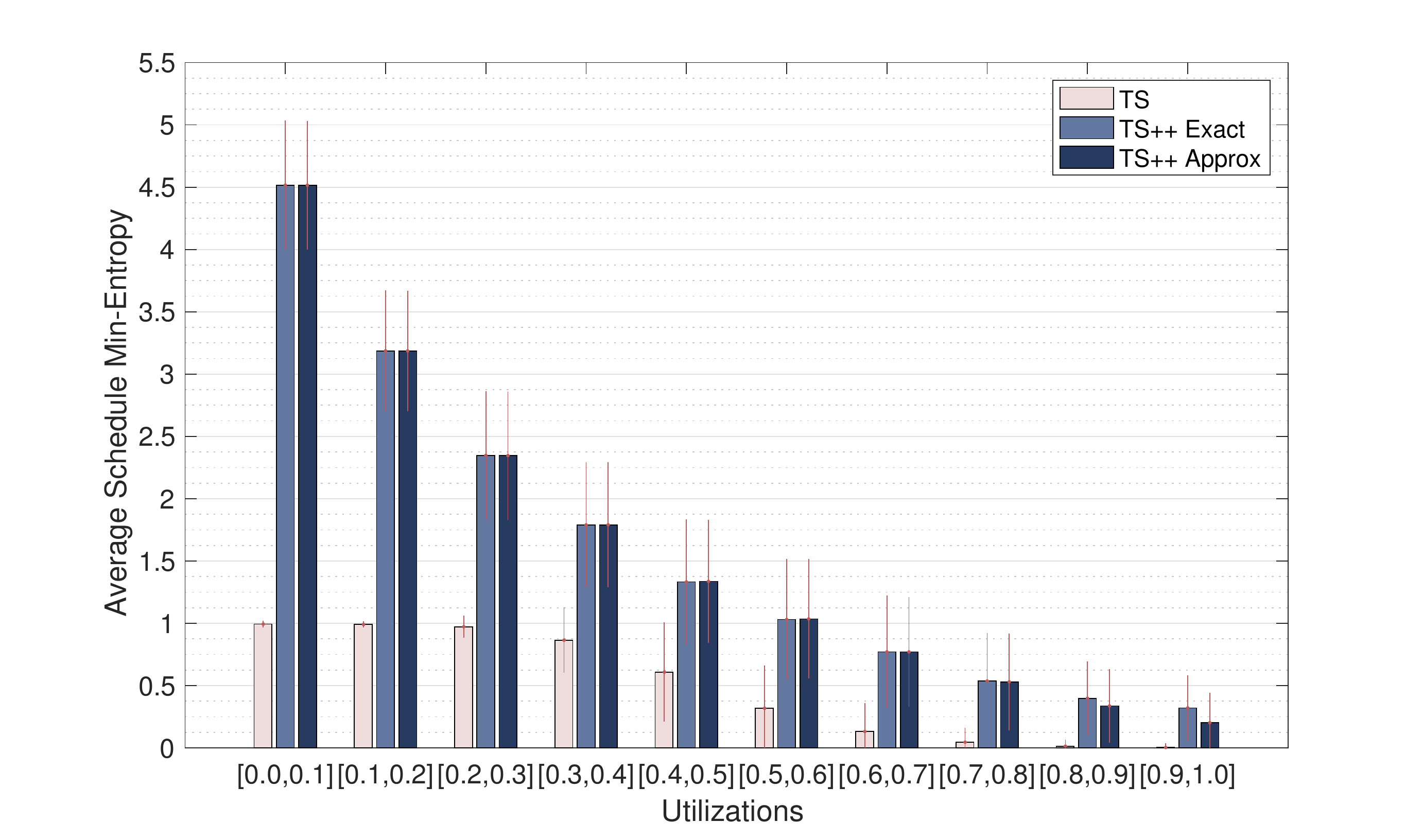}
\caption{The average schedule min-entropy when randomized by \textsc{TS} and \textsc{TS++}. A higher schedule min-entropy indicates a lower vulnerability to timing inference attack.}\label{fig:exp_new_entropy}
\end{figure}

Figure~\ref{fig:exp_new_entropy} compares the \emph{schedule min-entropy} of the three methods. It is interesting to note that, in contrast to the schedule entropy, the schedule min-entropy is high when the system utilization is low regardless of how the schedules are randomized. This is because such a task set is mainly composed of tasks with small utilization. Remind that the sum of $\Pr(x_t = \tau_i)$ over the period $p_i$ is $e_i$ as explained in Section~\ref{subsec:sched_min_entropy}. Accordingly, a small-utilization task inherently has a smaller probability of appearing in each time slot. For this very reason, the schedule min-entropy is low if tasks are likely to have high utilization. Figure~\ref{fig:exp_max_task_util_to_sched_min_entropy} shows how the \emph{maximum} task utilization of each task set affects its schedule min-entropy. As proven in Theorem~\ref{theorem:upper_sched_min_entropy}, the schedule min-entropy is upper-bounded by $-\log(\max u_i)$ and the result in the figure demonstrates that the schedule min-entropy is in relation to the maximum task utilization. Although $\max u_i$ is not an absolute indicator of the schedule min-entropy, we can see that a task set with small task utilization is in general less vulnerable to an adversary's correct prediction.

\begin{figure}[t]
  \centering
  \includegraphics[width=1\columnwidth]{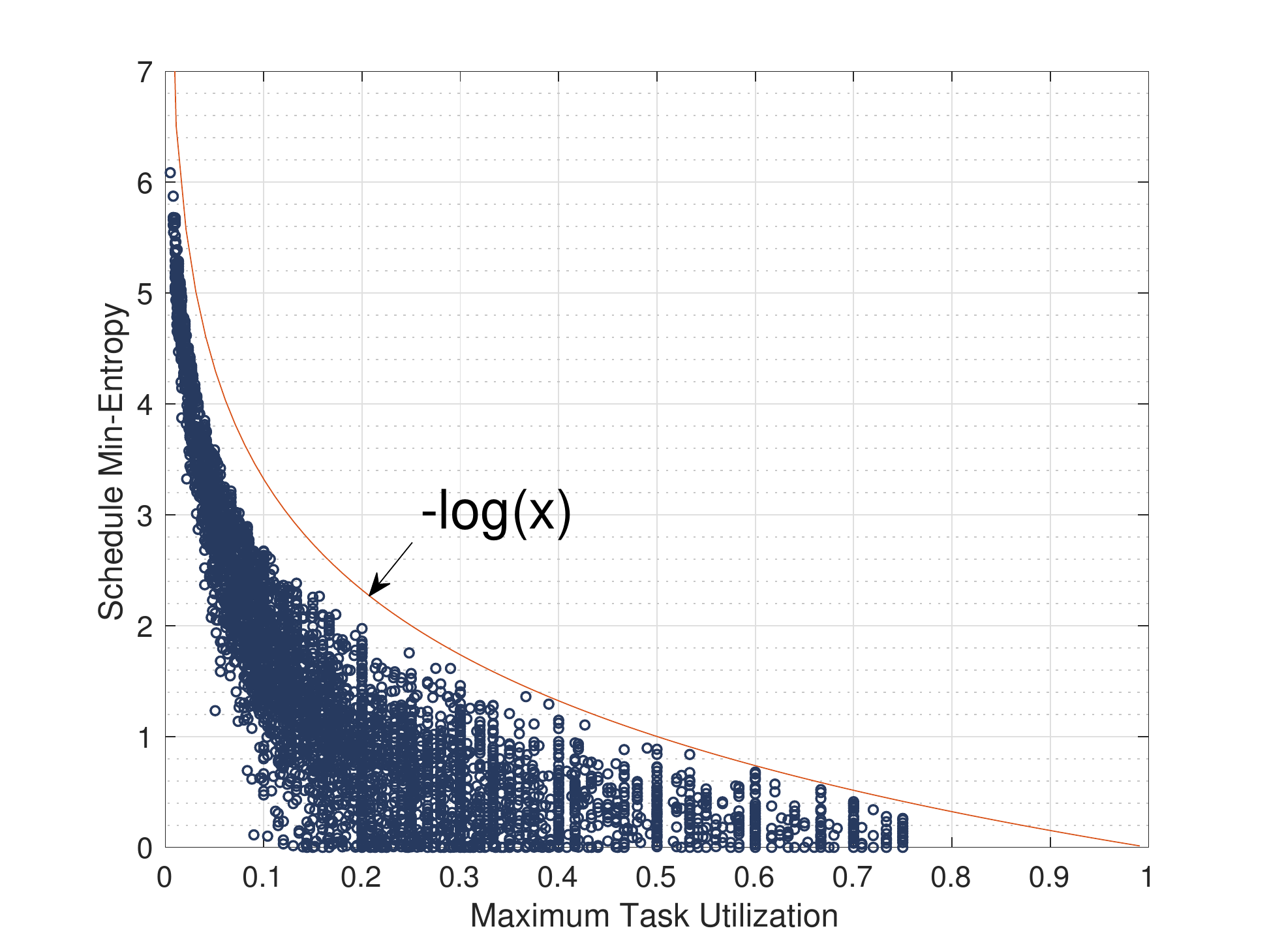}
  \caption{The per-task set maximum task utilization and the schedule min-entropy when randomized by \textsc{TS++ Approx}.}\label{fig:exp_max_task_util_to_sched_min_entropy}
  \end{figure}
Returning to the results in Figure~\ref{fig:exp_new_entropy}, we can see that the schedule min-entropy of \textsc{TS} is significantly lower compared to those of the \textsc{TS++} methods. Only few task sets (51 out of 6000 sets) have a schedule min-entropy above 1 under \textsc{TS}. This means that under the operation of \textsc{TS}, the chance that an adversary can make a correct guess on task execution at an arbitrary time is $50\%$ or higher in the best-case. This likelihood increases considerably with the system utilization. Table~\ref{table::exp_num_of_zero_entropy} shows the percentage of task sets that ended up having \emph{zero} schedule min-entropy; there exists at least one time slot on which the task execution is deterministic (i.e., $\Pr(x_t = \tau_i)=1$ for some $t$ and $\tau_i$). For \textsc{TS}, total 1354 sets have zero schedule min-entropy while \textsc{TS++ Approx} and \textsc{TS++ Exact} have 269 and 0 sets, respectively. It should be noted that, although rare, \textsc{TS++ Exact} could also result in zero schedule min-entropy depending on the task set's characteristic.

From these results we can see that \textsc{TS++ Approx} performs as well as \textsc{TS++ Exact} for most of the task sets, although they differ, albeit slightly, for high utilization group. This is caused by the approximations in the simulation of priority inversion in \eqref{eq:theorem_D1} and in the overflow $\overline{\rho}_{h,t,t'}$ in Section~\ref{subsec:approx_newTS} when testing if an inactive task $\tau_h$ would miss its deadline. These analyses become pessimistic, especially when computing the interference from higher priority tasks, as task utilizations increase. Accordingly, \textsc{TS++ Approx} becomes more conservative (thus randomizes less) when the system is highly packed. Figure~\ref{fig:v_h_ratio}, which shows the per-task set average ratio of $v_h$ to $\overline{V}_h$, supports this. We used the initial remaining budget that is set at job release as $v_h$. Recall that $\overline{V}_h$ is used for testing an inactive task's schedulability; if the predicted overflow exceeds it, priority inversion is not allowed by $lp(\tau_h)$. Thus, a high $v_h/\overline{V}_h$ means $\tau_h$ would have been able to allow more priority inversions than estimated. 

\begin{figure}[t]
  \centering
  \includegraphics[width=1\columnwidth]{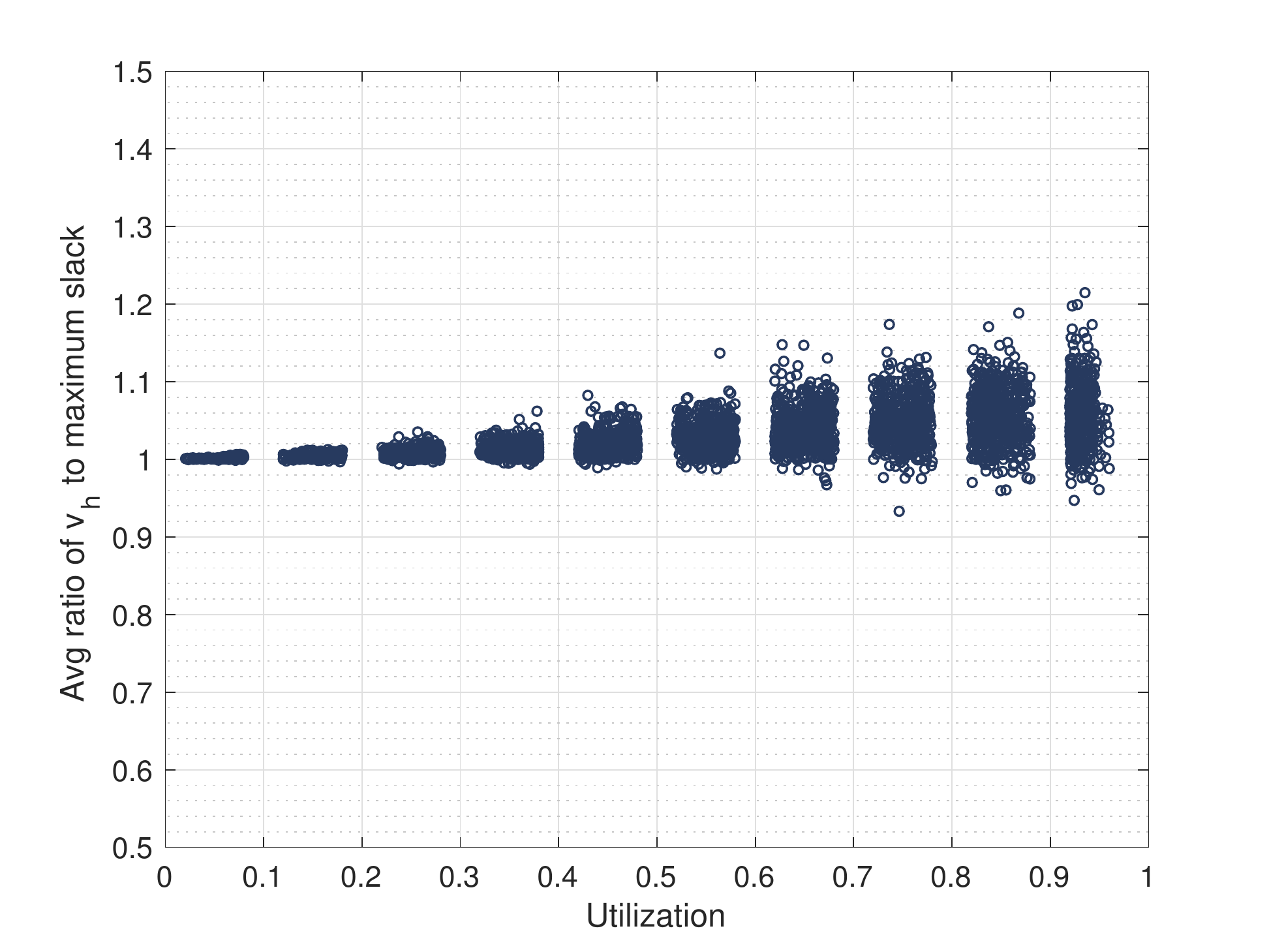}
  \caption{The per-task set average $v_h/\overline{V}_h$ in \textsc{TS++ Approx}. }\label{fig:v_h_ratio}
  \end{figure}  

An effective schedule randomization algorithm is desired to \emph{distribute} job executions over a wider range of time, otherwise it is easier for an adversary to target particular ranges to see job executions. Hence, we measured the execution range defined by the difference between the first and the last time slots where each task appears. The top plot in Figure~\ref{fig:exp_resp_time_range} shows the average ratio of the execution range to period in each task set when randomized by \textsc{TS} and \textsc{TS++ Approx}. As seen in the results, in contrast to \textsc{TS}, \textsc{TS++} achieves the ratios of 1 (except for one instance that achieved 0.995), which means that tasks can appear virtually \emph{anytime}. The reason why \textsc{TS} has narrower execution range is mainly because of the uniform job selection (not to mention the pessimism in inversion budget calculation) that leads tasks to have temporal locality as discussed in Section~\ref{subsec:weighted_selection}. In order to see the impact of the job selection process, we compare \textsc{TS++ Exact} with the weighted job selection against that with the uniform selection. As shown in the middle plot in Figure~\ref{fig:exp_resp_time_range}, it is the weighted job selection that makes task executions spread across a wider range. In addition, the comparison between \textsc{TS} and \textsc{TS++ Exact} with the uniform job selection indicates that the candidate selection process of \textsc{TS++} becomes more effective for system with high utilization. However, as the bottom plot in the figure shows, the gain obtained by the weighted job selection is offset by the pessimistic bounds on the inversion budgets in \textsc{TS} when the system has a high utilization. 

\begin{figure}[t]
  \centering
    \includegraphics[width=1\columnwidth]{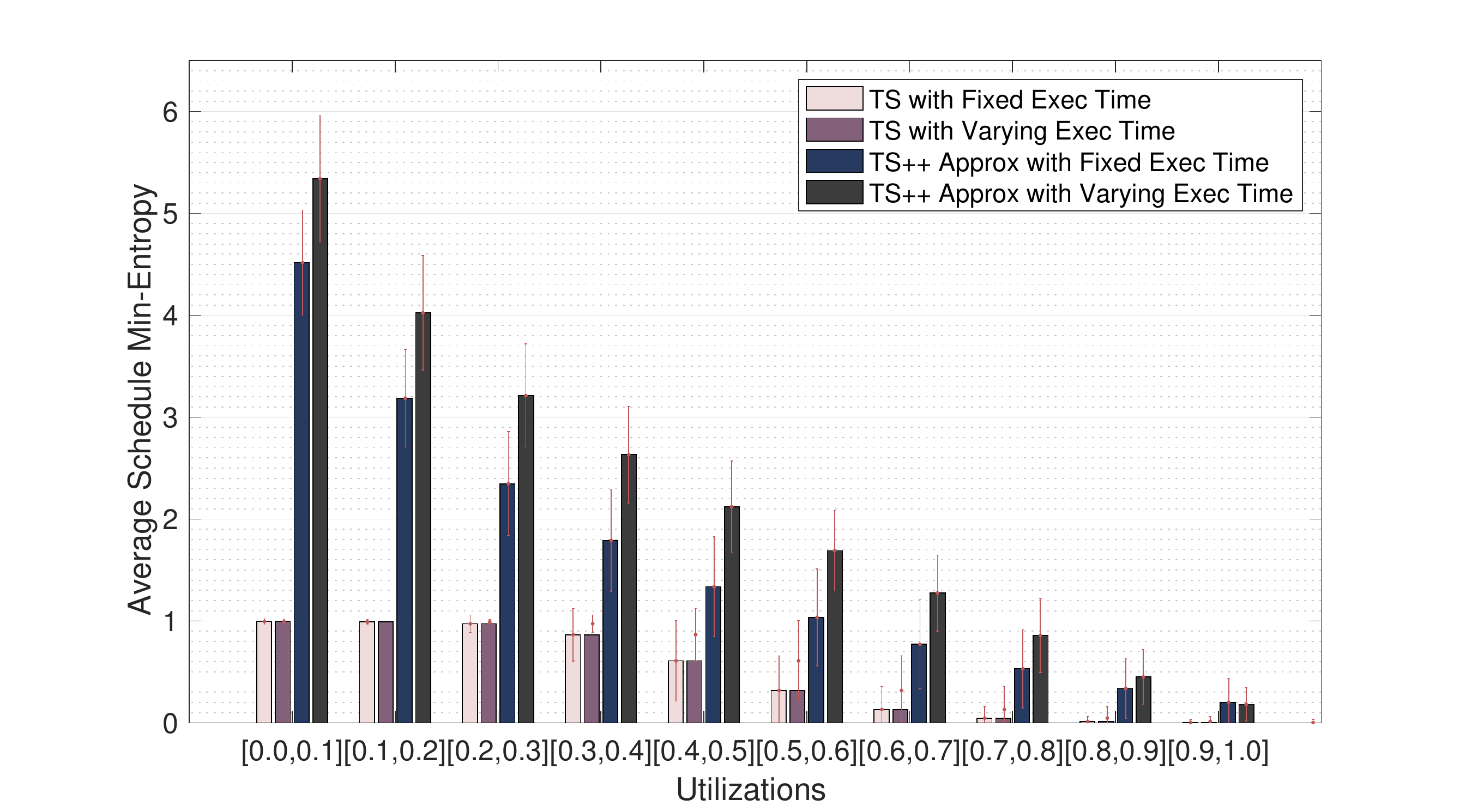}
  \caption{The average schedule min-entropy when job execution times are fixed or varying.}\label{fig:exp_varying_exec_time}
  \end{figure}
\revised{So far, the job execution times have been fixed to the worst-case execution times. As briefly mentioned earlier, varying the job execution times can naturally add more randomness to schedules. Figure~\ref{fig:exp_varying_exec_time} shows the results of \textsc{TS} and \textsc{TS++ Approx} -- the average scheduler min-entropy increases in both cases when job execution times are not fixed. For this experiment, we varied the execution time of each job invocation in a way that it is drawn from a uniform distribution over the interval $[70\%\cdot e_i, e_i]$ where $e_i$ is the worst-case execution time of $\tau_{i}$. Because the scheduler does not know what the execution time of each job invocation would be, all of the run-time computations of \textsc{TS++ Approx}, such as the overflow and maximum slack (See Section~\ref{subsec:approx_newTS}) still use the WCETs of tasks. The weighted job selection (Section~\ref{subsec:weighted_selection}) also uses the WCETs because the actual execution time is unknown until it finishes. The scheduler computes $\tilde{e}_{i,t}$, the residue execution time of $\tau_i$ at time $t$, by subtracting the amount of time $\tau_i$ has executed until $t$ from its WCET. Varying execution times can negatively affect the schedule randomness especially when the utilization is very high, as Figure~\ref{fig:exp_varying_exec_time} shows. This is mainly because jobs execute for shorter times than assumed by the scheduler. Hence, while a job could allow more priority inversions (because of its own shorter execution time and the higher-priority tasks' smaller interference), it may end up finishing earlier because of the over-estimation of busy period (thus under-estimation of the run-time slack) as a result of WCET-based estimations. The over-estimation grows with the utilization. }

\begin{figure}[t]
    \centering
    \includegraphics[width=1\columnwidth]{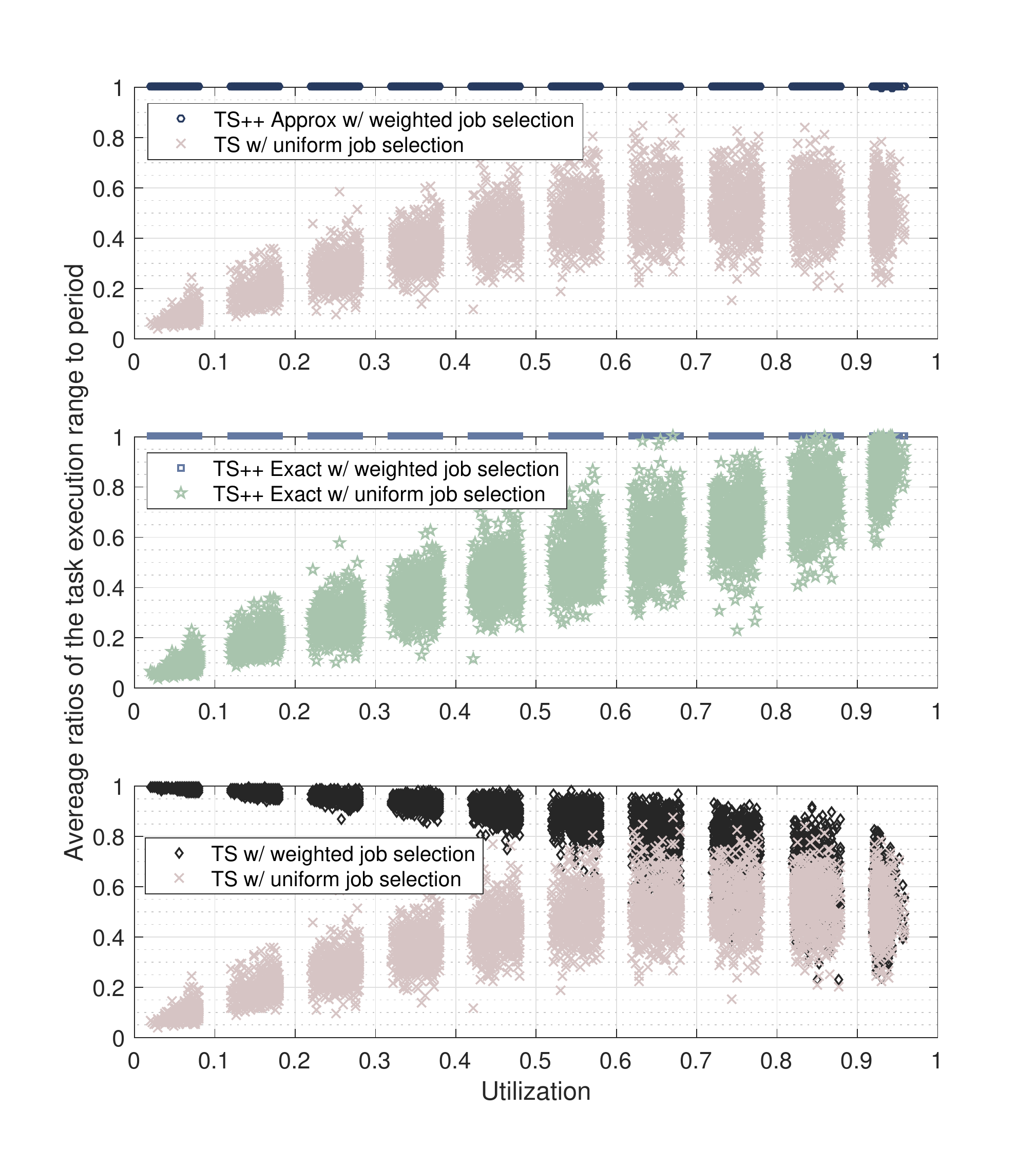}    
    \caption{The per-task set average ratio of the task execution range to period. }\label{fig:exp_resp_time_range}
    \end{figure}

Finally, we measured the number of context switches made by each task set to assess the cost associated with schedule randomness. On average, \textsc{TS++ Exact} causes 0.21\% more context switches than \textsc{TS++ Approx} (min=-0.23\%, max=6.79\%, stdev=0.61\%). \textsc{TS++ Exact} causes 47.65\% more context switches than \textsc{TS} on average (min=3.85\%, max=546.44\%, stdev=33.19\%). The extreme cases resulted from \textsc{TS}'s poor performance -- \textsc{TS} could not randomize schedules well, hence significantly fewer context switches and lower schedule entropy. Thus, the comparison of raw numbers of context switches is inadequate. Hence, we computed $\mathtt{EPS} = \frac{\mathtt{schedule\_min\_entropy}}{\mathtt{\#\_context\_switches}}$. Higher $\mathtt{EPS}$ means higher randomization at the same cost (or lower cost for the same level of randomization). On average, the $\mathtt{EPS}$ of \textsc{TS++ Exact} is 36.89\% higher than for \textsc{TS} (min=1.59\%, max=375.25\%, stdev=28.95\%).

\section{Related Work}
\label{related}

A number of studies have shown that real-time scheduling can leak information, whether intended or not. Son \emph{et al.} \cite{embeddedsecurity:son2006} showed that the rate monotonic scheduling is exposed to covert timing channel due to its scheduling timing constraints. Similarly, V\"{o}lp \emph{et al.} \cite{embeddedsecurity:volp2008} addressed the problem of information flows that can be established by altering scheduling behavior. The authors proposed modifications to the fixed-priority scheduler to close 
timing channels while achieving real-time guarantees. In \cite{embeddedsecurity:volp2013}, V\"{o}lp \emph{et al.} also tackled the issues of information leakage through shared resources such as real-time locking protocols (e.g., Priority Inheritance Protocol~\cite{Sha:1990:PIP}) and proposed transformation for them to prevent unintended information leakage \cite{embeddedsecurity:volp2013}. Similarly, architectural resources can be a source of unintended information flow. Mohan \emph{et al.} \cite{embeddedsecurity:mohan2014,Mohan:2016} addressed information
leakage through storage timing channels (e.g., shared caches) shared between real-time tasks with different security levels. The authors proposed a modification to the fixed-priority scheduling algorithm that cleans up shared storage (using flush mechanism) during a context-switch to a lower security-level task, and a corresponding schedulability analysis. This work was further extended to a generalized task model \cite{embeddedsecurity:pellizzoni2015} in which an optimal assignments of priority and task preemptibility considering security levels of tasks is also introduced. 

Chen \emph{et al.}~\cite{SchedLeakRTAS2019} demonstrated a timing inference attack against fixed-priority scheduling; an observer task can infer the timings (e.g., future arrival time) of certain tasks by observing its own execution intervals. Such attempts can be deterred by schedule randomization techniques such as \TS{} \cite{YoonTaskShuffler} and \NewTS{} presented in this paper. 
Kr\"uger \emph{et al.} \cite{Kruger2018ECRTS} proposed a randomization technique for time-triggered scheduling. To the best of our knowledge, these are the only studies that address real-time scheduling obfuscation against timing inference attacks, although randomization techniques can also be used for scheduling optimization 
(e.g., number of job completion \cite{Chrobak:2004}, CPU utilization \cite{perkovic2000randomization}). \revised{Nasri \emph{et al.} \cite{Nasri2019} proposed to use a conditional entropy to take into account the attacker's partial observation about the system. However, this requires a strong adversary model that the attacker can instantaneously observe the scheduler's state such as the ready queue and task schedule up to present.}

Randomization is a critical ingredient for moving target defense (MTD) techniques~\cite{okhravi2014finding}. Davi \emph{et al.}\cite{davi2013gadge} used address space layout randomization (ASLR) \cite{shacham2004} to randomize program code on-the-fly for each run to deter code-reuse attacks. Crane \emph{et al.} \cite{crane2015readactor} improved code randomization by enforcing execute-only memory to eliminate code leakage that allows an attacker to learn about the address space layout. Kc \emph{et al.} \cite{Kc:2003} took a finer-grained approach that creates a process-specific instruction set that is hard to be inferred by an adversary. Zhang \emph{et al.} \cite{zhang2014new} addressed a problem of information leakage through cache side-channels by randomly evicting cache lines and permuting memory-to-cache mappings. Jafarian \emph{et al.} \cite{Jafarian:2012} considered MTD in software defined networking (SDN) in which the controller randomly assigns (virtual) IP addresses to hosts in order to hinder adversaries from discovering targets. 

\section{Conclusion}
\label{sec:conclusion}

In this paper, we have presented \NewTS{} as a solution to raise the bar against timing inference attacks. By increasing timing uncertainty of real-time tasks, \NewTS{} poses a significant obstacle to adversaries' efforts to predict when tasks would execute, thus reducing their window of opportunity. We have shown that \NewTS{} increases the worst-case security of real-time schedules \emph{in a measurable way}; we did this by introducing a notion of schedule min-entropy, which captures an adversary's best chance of successful inference of task identity.  This enables a quantitative comparison of different task sets or different schedule randomization algorithms. Such an information-theoretic view opens up interesting questions on the security of real-time scheduling, such as the bandwidth of covert channels and mutual information between task timings, which we intend to investigate.

\end{document}